\documentclass{ims9x6}

\usepackage{cite,bbm,stmaryrd}

\DeclareMathAlphabet{\vecfont}{OT1}{cmr}{bx}{it}
\renewcommand{\vec}[1]{\vecfont{#1}}
\newcommand{\grad}{\boldsymbol{\nabla}}
\newcommand{\scaledmath}[2]{\scalebox{#1}{$\displaystyle{#2}$}}
\newcommand*{\arvec}[2][-2ex]{\begin{array}[b]{@{}c@{}}%
    \scaledmath{0.8}{\shortrightarrow}\\[#1]{#2}\end{array}}
\newcommand*{\arsinh}{\mathrm{arsinh}}
\newcommand*{\dbar}{{\mathchar'26\mkern-12mu\mathrm{d}}}
\newcommand*{\oper}[1]{\mathalpha{\textsf{#1}}}
\newcommand*{\1}{\oper{1}}
\newcommand*{\rl}{{\rangle\langle}}
\newcommand*{\cz}{\mathbb{C}}
\newcommand*{\gz}{\mathbb{Z}}
\newcommand*{\rd}{\mathrm{d}}
\newcommand*{\ri}{\mathrm{i}}
\newcommand*{\cA}{\mathcal{A}}
\newcommand*{\cC}{\mathcal{C}}
\newcommand*{\cD}{\mathcal{D}}
\newcommand*{\cE}{\mathcal{E}}
\newcommand*{\cF}{\mathcal{F}}
\newcommand*{\cI}{\mathcal{I}}
\newcommand*{\cN}{\mathcal{N}}
\newcommand*{\cR}{\mathcal{R}}
\newcommand{\cS}{\mathcal{S}}
\newcommand*{\tf}{T^\mathrm{TF}}
\newcommand*{\TF}{\mathcal{T}^\mathrm{TF}}
\newcommand*{\V}{\mathcal{V}}
\newcommand*{\W}{\mathcal{T}^\mathrm{W}}
\newcommand*{\X}{\mathcal{X}}
\newcommand*{\ga}{\vec{a}}
\newcommand*{\gj}{\vec{j}}
\newcommand*{\gx}{\vec{x}}
\newcommand*{\gy}{\vec{y}}
\newcommand*{\gp}{\vec{p}}
\newcommand*{\gu}{\vec{u}}
\newcommand*{\gK}{\vec{K}}
\newcommand*{\gR}{\vec{R}}
\newcommand*{\geta}{\boldsymbol{\eta}} %% added
\newcommand*{\gxi}{\boldsymbol{\xi}} %% added 
\newcommand*{\gq}{\mathfrak{q}}
\newcommand*{\gB}{\mathfrak{B}}
\newcommand*{\gF}{\mathfrak{F}}
\newcommand*{\gH}{\mathfrak{H}}
\newcommand*{\gS}{\mathfrak{S}}
\newcommand*{\nz}{\oper{N}}
\newcommand*{\rz}{\mathbb{R}}
\DeclareMathOperator*{\tr}{tr}
\DeclareMathOperator*{\sgn}{sgn}
\newcommand*{\gtf}{{\gamma_{\mathrm{TF}}^{\ }}}
\newcommand*{\fitf}{{\varphi_Z}}
\newcommand*{\const}{C}
\makeatletter
\renewcommand{\ps@plain}{%
  \renewcommand{\@oddhead}{\hfil\footnotesize%
    \raisebox{30pt}[0pt][0pt]{\parbox{300pt}{\centering%
      A contribution to the Proceedings of the\\{}%
      Workshop on Density Functionals for Many-Particle Systems\\{}
      2--29 September 2019, Singapore}}\hfil}%
  \renewcommand{\@evenhead}{\@oddhead}%
  \renewcommand{\@oddfoot}{\hfil\footnotesize%
        \raisebox{-8pt}[0pt][0pt]{\thepage}\hfil}%
  \renewcommand{\@evenfoot}{\@oddfoot}%
}
\makeatother
\renewcommand{\trimmarks}{\relax}
\begin{document}

\chapter{\uppercase{Mathematical Elements\\%
    of Density Functional Theory}}
\markboth{Heinz Siedentop}%
         {Mathematical elements of density functional theory}

\author{Heinz Siedentop}
\address{Mathematisches
  Institut, Ludwig-Maximilans Universit\"at M\"unchen\\
  Theresienstr.~39, 80333 M\"unchen, Germany,\\ and\\
  Munich Center for
  Quantum Science and Technology (MCQST)\\
  Schellingstr.~4, 80799
  M\"unchen, Germany\\[1ex]
  e-mail: h.s@lmu.de}

\begin{abstract}
  We review some of the basic mathematical results about
  density functional theory. 
\end{abstract}

\section{Introduction \label{Einleitung}}
Already in classical mechanics it turned out that the microscopic
equation describing many particles interacting via the Coulomb force
cannot be solved analytically. Approximation schemes are
indispensably. The idea of describing such systems by effective
equations depending only on few variables of interest, like the
density of matter and its velocity field has a long history that dates
back to at least to Euler \cite{Euler1757e,Euler1757,Euler1757C} long
before the advent of quantum mechanics. In quantum mechanics of atoms
and other fermionic systems the need for such an effective description
was also immediately recognized. Only two years after Heisenberg's
\cite{Heisenberg1925} fundamental discovery Thomas \cite{Thomas1927}
and Fermi \cite{Fermi1927,Fermi1928} introduced the first --- in
today's language --- density functional theory which is now known as
Thomas--Fermi theory. Although relatively simple it is not only of
historical importance. As we will see in Section \ref{tftheorie} it
becomes asymptotically exact for heavy atoms and serves as a
mathematical tool to prove fundamental features of matter like its
stability and, closely related, the existence of the thermodynamic
limit of Coulomb systems. Later, various corrections of the theory were
made, e.g., Dirac \cite{Dirac1930} took the exchange energy into
account and Weizs\"acker \cite{Weizsacker1935} inhomogeneities of the
electron gas. A detailed overview over these early developments is
offered by Gombas \cite{Gombas1949}. The subject received an immense
push with an observation of Hohenberg and Kohn
\cite{HohenbergKohn1964}. They argued that there is at most one
external potential for a system of $N$ fermions interacting via
Coulomb potentials for which a given one-particle density with
particle number $N$ is its ground state density. They concluded from
this that there is a functional of the form
\begin{equation}
  \label{hk}
  \cE^{\mathrm{HK}}(\rho)+\int_{\rz^3}\rd\gx\, V(\gx)\rho(\gx)
\end{equation}
whose minimizer is a reduced one-particle electronic ground state
density of the Hamiltonian
\begin{equation}
  \label{H}
  H_{V,N}:=\sum_{n=1}^N\biggl(-\frac12\Delta_n+V(\gx_n)\biggr)
  +\sum_{1\leq m<n\leq N}\frac{1}{|\gx_m-\gx_n|}.
\end{equation}
(As usual the colon in connection with an equality sign indicates a
definition where the defined quantity is on the side of the colon.)
The functional $\cE^\mathrm{HK}$ depends only on the kinetic energy
operator and the interaction energy of the electrons which are assumed
to be the nonrelativistic kinetic energy and the Coulomb force. Both
choices, namely nonrelativistic kinetic energy and Coulomb
interaction, are made only for convenience and definiteness at this
point. Although Hohenberg and Kohn merely claimed the existence of
such a functional without offering any construction of
$\cE^\mathrm{HK}$, it nevertheless triggered a tsunami of results on
the subject that is still gathering momentum sixty years after its
publication.

\section{Elements of quantum mechanics of electrons}
\subsection{Electronic Hilbert spaces}
We write $\gH_1$ for the one-electron Hilbert space.
The $N$-electron Hilbert space is its antisymmetric tensor product
\begin{equation}
  \label{eq:hn}
  \gH_N:= \bigwedge_{n=1}^N\gH_1,
\end{equation}
and finally the Fock space of electrons is
\begin{equation}
  \label{eq:f}
  \gF:=\bigoplus_{N=0}^\infty\gH_N
\end{equation}
with the understanding that $\gH_0:=\cz$. The summand $\gH_0$ is also
called the vacuum space.

To be concrete we focus on the case of nonrelativistic quantum
mechanics where $\gH_1=L^2(\Gamma)$ is the space of functions $\psi$
of the space-spin variable $x:=(\gx,\sigma)\in\Gamma:=\rz^3\times\{1,2\}$ with
finite scalar product $(\psi,\psi)<\infty$. Here
\begin{equation}
  \label{eq:sp1}
  (\psi,\tilde\psi):=\int_{\Gamma}\rd x \; \overline{\psi(x)}\tilde\psi(x):=
  \sum_{\sigma=1}^2\int_{\rz^3}\rd\gx\;
  \overline{\psi(\gx,\sigma)}\tilde\psi(\gx,\sigma)
\end{equation}
using the notation $\int_\Gamma\rd x :=\int_{\rz^3}\rd \gx \sum_{\sigma=1}^2$.

The $N$-electron Hilbert space $\gH_N$ is the space of antisymmetric
functions $\psi$ of space-spin variables with $(\psi,\psi)<\infty$
where
\begin{equation}
  \label{eq:spN}
  (\psi,\tilde\psi)_N :=\int_{\Gamma^N}\rd x_1\cdots\rd x_N \;
  \overline{\psi(x_1,\ldots,x_N)}\tilde\psi(x_1,\ldots,x_N).
 \end{equation}

 Eventually $\gF$ is the space of all sequences
 $\psi:=(\psi_0,\psi_1,\psi_2,\ldots)$ with $\psi_N\in\gH_N$ such that
 $(\psi,\psi)_\gF<\infty$ with
 \begin{equation}
   \label{eq:spF}
   (\psi,\tilde\psi)_\gF:=\sum_{N=0}^\infty(\psi_N,\tilde\psi_N)_N.
 \end{equation}
 In the following we will drop any indices with scalar products.

 \subsection{Electronic states}
 From an abstract point of view, a state $\omega$ is simply a continuous
 linear functional on 
 the bounded operators $\gB(\gH)$ of the underlying Hilbert space
 $\gH$, which is also positive and normalized, i.e.,
 \begin{itemize}
   \item  For all $A\in\gB(\gH)$ we have $\omega(A^*A)\in\rz_+$ (positivity).
   \item $\omega(1)=1$ (normalization).
 \end{itemize}
 
 Given any $f,g\in \gH$ we will use the physics notation
 $|g\rangle\langle f|$ for the operator
 \begin{equation}
   \label{eq:bra}
   \begin{split}
     \gH&\to\gH,\\
     h&\mapsto (f,h)g.
   \end{split}
 \end{equation}
 In particular, if $f$ is normalized, then $|f\rangle\langle f|$ is the
 orthogonal projection onto the one-dimensional subspace spanned by
 $f$.  Given any set of orthonormal vectors $\xi_1,\xi_2,\ldots\in\gH$
 and nonnegative weights $w_1,w_2,\ldots\in\rz_+$ with
 $w_1+w_2+\cdots=1$
\begin{equation}
  \label{eq:dm}
  d:= w_1|\xi_1\rl\xi_1|+w_2|\xi_2\rl\xi_2|+\cdots
\end{equation}
is called a density matrix. To each density matrix one can associate
in natural way a state $\omega_d$ via the relation
\begin{equation}
  \label{eq:gemischt}
  \omega_d(A):=\tr(Ad).
\end{equation}
If there is a single normalized vector $\psi\in \gH_N$ such that
$d=|\psi\rl\psi|$ then $\omega_{|\psi\rl\psi|}$ is called a pure state
which --- in abuse of notation --- is also written for brevity as
$\rho_\psi$. We have
\begin{equation}
  \label{eq:rein}
  \omega_\psi(A):=\omega_{|\psi\rl\psi|}(A)= \tr(A|\psi\rl\psi|)=(\psi,A\psi).
\end{equation}
Because of the above relations, it is customary --- although strictly
speaking abusing notation again --- to address normalized vectors in a
Hilbert space also as states, in fact as pure states, and density
matrices with rank larger than one as mixed states.

The Hilbert spaces of relevance for us will be the Fock space $\gF$
and its summands, i.e., the $N$-electron spaces $\gH_N$.

\subsection{Creation and annihilation operators\label{ca}}
For our discussion it is handy to use creation and annihilation
operators.  Given $f\in\gH_1$ and $\psi=(\psi_0,\psi_1,\ldots)\in\gF$ we
define $a(f):\gF\to\gF$ component-wise --- abusing notation writing
$\psi_N$ instead of $(0,\ldots,\psi_N,0,\ldots)$ --- by
\begin{align}
  \label{eq:a}
  &\mathrel{\phantom{=}}[a(f)(\psi_N)](x_1,\ldots,x_{N-1})\nonumber\\
  &:= \begin{cases}
    \sqrt{N}\scaledmath{0.9}{\int_{\Gamma}} \rd x\,
    \overline{f(x)}\psi_N(x,x_1,\ldots,x_{N-1})&\text{for\ } N\geq 1,\\
    0&\text{for\ } N=0.
    \end{cases}
\end{align}
  Three properties are obvious for all $f\in\gH_1$:
  \begin{itemize}
  \item the family $a$ is conjugate linear in $f$,
  \item the operators $a(f)$ are bounded and linear on $\gF$,
  \item the operators $a(f)$ map the $N$-particle sector of the Fock
    space to its $(N-1)$-particle sector (for all $N\geq1$).
  \end{itemize}
  The adjoint operators $a^*(f)$ are called creation operators. We
  claim that they are also given component-wise by the formula
  \begin{align}
    \label{eq:a*}
    &\mathrel{\phantom{=}}[a^*(f)(\phi)](x_0,\ldots,x_{N-1})\nonumber\\
    &:= \frac{1}{\sqrt{N}}\sum_{n=0}^{N-1}(-1)^nf(x_n)
    \phi(x_1,\ldots,\hat x_n,\ldots,x_{N-1})
  \end{align}
  where the hat indicates the omission of the variable below it.
  Obviously these operators are also bounded, linear, and map $\gH_N$
  to $\gH_{N+1}$.  That \eqref{eq:a*} gives indeed the adjoint
  operator can be seen as follows. Suppose $\phi\in\gH_{N-1}$ and
  $\psi\in\gH_N$ then
  \begin{align}
    &\mathrel{\phantom{=}}\bigl(\phi,a(f)\psi\bigr)\nonumber\\
    &=\sqrt{N}
      \int_{\Gamma^{N-1}}\rd x_1\cdots \rd x_{N-1}\,
      \overline{\phi(x_1,\ldots,x_{N-1})}
    \int_\Gamma\rd x_0\,
    \overline{f(x_0)}\psi(x_0,\ldots,x_{N-1})\nonumber\\
    &=\sqrt{N}
      \int_{\Gamma^{N}}\rd x_0\cdots \rd x_{N-1}\,
      \overline{f(x_0)\phi(x_1,\ldots,x_{N-1})}
   \psi(x_0,\ldots,x_{N-1})\nonumber\\ 
   &=
     \int_{\Gamma^{N}}\rd x_0\cdots \rd x_{N-1}\,
     \underbrace{\sum_{n=0}^{N-1}\frac{(-1)^n}{\sqrt{N}}\overline{f(x_n)\phi(x_1,\ldots,\hat x_n,\ldots,x_{N-1})}}_{=:\overline{[a^*(f)(\phi)](x_0,\ldots,x_{N-1})}}\nonumber\\&\rule{90pt}{0pt}\times
    \psi(x_0,\ldots,x_{N-1}).
  \end{align}

  For given orthonormal basis $\xi_1,\xi_2,\ldots$ of $\gH_1$
  the short hand
\begin{equation}
  \label{eq:an}
  a^*_n:=a^*(\xi_n),\ a_n:=a(\xi_n)
\end{equation}
is customary. One even extends these to ``eigenstates'' of position
by writing $a^*_x:=a^*(\delta_\gx\delta_{\cdot,\sigma})$ or of momentum
by writing
$a^*_p:=a^*(\exp(\ri \gp\cdot
\cdot)\delta_{\cdot,\sigma}/(2\pi)^\frac32)$ and of other
observables. These expressions become meaningful when ``integrated''
against a suitable test function in the sense of a distribution, e.g.,
we have
\begin{equation}
a^*(f)=:\int_\Gamma\rd x\, f(x)a^*_x,\ a^*\bigl(\cF(f)\bigr)
=:\int_\Gamma\rd p \,f(p)a_p^*
\end{equation}
where we write $\cF(f)$ for the Fourier transform of $f$. 
It is also common to write them as
\begin{equation}
  \label{eq:axa*x}
  a_x=\sum_n\xi_n(x)a(\xi_n),\ a_x^*=\sum_n\overline{\xi_n(x)}a^*(\xi_n)
\end{equation}
(see, e.g., Schweber \cite[Chapter 6.e, Formulae (63) and
(64)]{Schweber1961}) using Parseval's identity.

Eventually we mention that the creation and annihilation operators
fulfill the canonical anticommutation relations (Jordan and Wigner
\cite[Formulae (36) and (40)]{JordanWigner1928}
\begin{align}
  \label{eq:car}
    &a(f)a(g)+a(g)a(f)=a^*(f)a^*(g)+a^*(g)a^*(f)=0,\nonumber\\
    &a^*(f)a(g)+a(g)a^*(f) = (g,f).
\end{align}

\subsection{Reduced densities and density matrices\label{rD}}

Given a state $\omega\in\gB(\gF)'$ and $f,g\in\gH_1$ we define the
sesquilinear form
\begin{align}
  \label{eq:q}
    \gq_\omega: \gH_1\times\gH_1&\to \cz,\nonumber\\
    (f,g)&\mapsto  \omega\bigl(a^*(f)a(g)\bigr).
\end{align}
We have
$0 \leq \omega\bigl(a^*(f)a(f)\bigr)\leq
\omega\bigl(a^*(f)a(f)\bigr)+\omega\bigl(a(f)a^*(f)\bigr)=(f,f)$,
i.e., $\gq_\omega$
defines a positive linear operator $\gamma_\omega$ bounded from above
by one. Moreover $\gamma_\omega$ is trace class, and, if $\omega$
lives on $\gB(\gH_N)$ only, then $\tr(\gamma_\omega)=N$.

It is enough to show the trace class property for the case that
$\omega$ is a pure $N$-electron state $\psi\in \gH_N$. Then
\begin{align}
  \label{matrixgamma}
  &\mathrel{\phantom{=}}\gq_{|\psi\rl\psi|}(f,g)
    := \tr\bigl(a^*(f)a(g)|\psi\rl\psi|\bigr)=\bigl(a(f)\psi,a(g)\psi\bigr)
    \nonumber\\
  &= \int_{\Gamma^{N-1}}\rd x_1\cdots \rd x_{N-1}\,
  \overline{a(f)\psi(x_1,\ldots,x_{N-1})}a(g)\psi(x_1,\ldots,x_{N-1})\nonumber\\
  &=\int\limits_\Gamma\rd x\int\limits_\Gamma\rd y\,\overline{g(x)}\nonumber\\
  &\rule{20pt}{0pt}\times
  \underbrace{N\int\limits_{\Gamma^{N-1}}\rd x_1\cdots \rd x_{N-1}\,
    \psi(x,x_1,\ldots,x_{N-1})\overline{\psi(y,x_1,\ldots,x_{N-1})}}
    _{=:\gamma_\psi(x,y)}f(y)\nonumber\\
&=\int_\Gamma\rd x\, \overline{g(x)}\int_\Gamma\rd y\,\gamma_\psi(x,y)f(y)
=(g,\gamma_\psi f).
\end{align}
Since $0\leq \gq_\omega(f,f)\leq (f,f)$, we have also that
$0\leq\gamma_\omega\leq1$. Because of the positivity of $\gamma_\omega$,
it suffices to show that its trace exists. By \eqref{matrixgamma} we
have for an orthonormal basis $\xi_1,\xi_2,\ldots$
\begin{align}
  &\mathrel{\phantom{=}} \tr(\gamma_\omega)
    = \sum_{n=1}^\infty \gq_{|\psi\rl\psi|}[\xi_n]\nonumber\\
  & = \sum_{n=1}^\infty\int\limits_\Gamma\rd x
    \int\limits_\Gamma\rd y\,\overline{\xi_n(x)}\nonumber\\
  &\rule{25pt}{0pt}\times N\!\int\limits_{\Gamma^{N-1}}\!\rd x_1
    \cdots \rd x_{N-1}\,
    \psi(x,x_1,\ldots,x_{N-1})\overline{\psi(y,x_1,\ldots,x_{N-1})}
    \xi_n(y)\nonumber\\
  & =\int\limits_\Gamma\rd x\, N\int\limits_{\Gamma^{N-1}}\rd x_1
    \cdots \rd x_{N-1}
    \,\psi(x,x_1,\ldots,x_{N-1})\overline{\psi(x,x_1,\ldots,x_{N-1})}
    =N
\end{align}
  where we used Parseval's identity to arrive at the last line.

The operator $\gamma_\omega$ is called the \textit{one-particle
  reduced density matrix of} $\omega$, $\gamma_\omega(x,y)$ is its
integral kernel, and the (spin summed) \textit{one-particle reduced
  density} is
\begin{equation}
  \label{dichte}
  \rho_{\gamma_\omega}(\gx):= \sum_\sigma\sum_n\lambda_n|\xi_n(x)|^2
\end{equation}
where the $\{\xi_1,\xi_2,\ldots\}$ is a complete orthonormal set of
eigenvectors of $\gamma_\omega$, known as natural orbitals, and the
$\lambda_n$ are the corresponding eigenvalues. Formally
$\rho_\omega(\gx)$ is the spin-summed diagonal
$\sum_\sigma\gamma_\omega(\gx,\sigma,\gx,\sigma)$ of the reduced
one-particle density matrix. Using the above notation, we have
\begin{equation}
  \label{eq:gxy}
  \omega(a_x^{\phantom{*}}a^*_y)= \gamma_\omega(x,y)
\end{equation}
for the kernel of $\gamma_\omega$.

This generalizes to $k$ particles: the $k$-particle reduced density
matrix $\gamma_\omega^{(k)}$ of a state $\omega$ has the integral
kernel
\begin{equation}
  \gamma_\omega^{(k)}(x_1,\ldots,x_k,y_1,\ldots,y_k)
  :=\frac{1}{k!}\omega(a_{x_k}\cdots a_{x_1}a^*_{y_1}\cdots a^*_{y_k}). 
\end{equation}
The $k$-particle reduced density
\begin{equation}
  \rho^{(k)}_\omega(\gx_1,\ldots,\gx_k)=\sum_{\sigma_1,\ldots,\sigma_k}^k
  \gamma_\omega^{(k)}(x_1,\ldots,x_k,x_1,\ldots,x_k)
\end{equation}
is the spin-summed diagonal of the $k$-particle density matrix.

\subsection{Observables}
\subsubsection{The number operator}
The space
\begin{equation}
Q_{\nz}:=\Biggl\{(\psi_0,\psi_1,\ldots)\in\gF\Biggm| \sum_{N=1}^\infty
N\int_{\Gamma^N}\rd x\, |\psi_N(x)|^2<\infty \Biggr\}
\end{equation}
is dense in $\gF$ and the number of electrons in a state
$\psi=(\psi_0,\psi_1,\ldots)$ is given by the quadratic form
\begin{equation}
  \label{anzahl}
  \gq_{\nz}[\psi]:=\int_\Gamma\rd x\, (a_x\psi,a_x\psi)
  =  \sum_{N=1}^\infty N\int_{\Gamma^N}\rd x\, |\psi_N(x)|^2.
\end{equation}
It is well defined on $Q_\nz$, closed, and bounded from below. The
associated selfadjoint operator defined according to Friedrichs is
called the number operator $\nz$. In a common abuse of notation it is
written as $\nz=\int_\Gamma\rd x\, a^*_xa_x^{\phantom{*}}$.

\subsubsection{The electronic Hamiltonian\label{hamiltonian}}
The one-particle kinetic energy operator is a positive, selfadjoint,
and translation-invariant operator which in nonrelativistic quantum
mechanics is $-\tfrac12\gp^2\otimes\1_{\cz^2}$ with $\gp:=-\ri\grad$.
We should write $V\otimes\1_{\cz^2}$ for the
one-particle external potential when spin independent; but as usual we
will omit factors that are one. We will assume that the kinetic energy
controls the one-particle potential, technically
$\exists_{a\in[0,1)}\exists_{M\in\rz}\forall_{f\in H^1(\Gamma)}$
\begin{equation}
  \label{relbesch}
\left|\int_\Gamma\rd x\, |f(x)|^2V(\gx)\right|\leq
\int_\Gamma\rd x\,{\left(\frac a2|\nabla f(x)|^2 +M |f(x)|^2\right)}.
\end{equation}
Here
\begin{equation}
  H^1(\Gamma):= \biggl\{f\in L^2(\Gamma)\biggm|
  \int_\Gamma\rd \xi \,|\xi \cF(f)(\xi)|^2<\infty\biggr\},
\end{equation}
i.e., a state $f$ is in $H^1(\Gamma)$ if it has finite kinetic
energy. The space $H^1$ is called the Sobolev space of order one. It
naturally characterizes all states that have a finite energy in
nonrelativistic quantum mechanics.

We write
\begin{equation}
  Q_{\oper{H}}
  :=\Biggl\{(\psi_0,\psi_1,\ldots)\in\gF\Biggm|\sum_{N=1}^\infty
  \int_{\Gamma^N}\rd \xi
      \,|\xi\cF(\psi_N)(\xi)|^2<\infty \Biggr\}
\end{equation}
for the space of all $\psi\in\gF$ which have finite kinetic
energy. The energy $\cE(\psi)$ of the electrons in a state
$\psi\in Q_{\oper{H}}$,
\begin{align}
  \label{eq:ham}
  &\mathrel{\phantom{=}}\cE_V[\psi]\nonumber\\
  &:=\int_\Gamma \rd x \left( \frac12\|\gp_\gx a_{x}\psi\|^2
    + V(\gx) \| a_x\psi\|^2\right)
    +  \frac12\int\limits_\Gamma \rd x \int\limits_\Gamma \rd y\,
    \frac{(\psi, a^*_xa^*_ya_ya_x\psi)}{|\gx-\gy|}\nonumber\\
  &=\sum_{N=1}^\infty \,\int\limits_{\Gamma^N}\rd x\,
    \Biggl[\binom{N}{1}
    {\left( \frac12|\grad_{\gx_1}\psi_N(x)|^2+V(\gx_1)|\psi_N(x)|^2\right)}
  +\binom{N}{2}
  \frac{|\psi_N(x)|^2}{|\gx_1-\gx_2|}\Biggr]\nonumber\\
  &=\sum_{N=1}^\infty \,\int\limits_{\Gamma^N}\rd x\,\Biggl[
    \sum_{n=1}^N\left( \frac12|\grad_{\gx_n}\psi_N(x)|^2
    +V(\gx_n)|\psi_N(x)|^2\right)\nonumber\\
   &\rule{70pt}{0pt}\mbox{}
   +\sum_{1\leq m<n\leq N}\frac{|\psi_N(x)|^2}{|\gx_m-\gx_n|}\Biggr],
\end{align}
is well defined and $\cE_{V}|_{\gH_N}$ is closed and bounded from
below by $-MN$. The selfadjoint operator associated to it according
to Friedrichs is the $N$-electron Hamiltonian $H_{V,N}$ giving meaning
to \eqref{H} as a selfadjoint operator. The direct sum of these
Hamiltonians is the second quantized Hamiltonian
$\oper{H}$. Furthermore, we remark that, e.g., in the case of a
Coulomb potential $V(\gx)=-Z/|\gx|$, i.e., the standard atomic
Hamiltonian,
\begin{equation}
  \cE_{-Z/|\cdot|}[\psi]\geq
  \inf\bigl\{\sigma(H_{-Z/|\cdot|,2Z+1})\|\psi\|^2\bigr\}
\end{equation}
which is a consequence of the fact that there are no arbitrarily
negative ions. In fact, Lieb \cite{Lieb1984} showed that less than
$2 Z+1$ electrons can be bound. (We will exhibit the argument leading to
this bound in the proof of Theorem \ref{rafw} in the context of the
Thomas--Fermi--Weizs\"acker functional). In this case, $\cE_V$ is bounded
from below and the Friedrichs extension in the entire Fock
space is directly possible.

We finish the section by rewriting \eqref{eq:ham} in terms of the
one-particle density matrix $\gamma_\psi$ and the two-particle density
$\rho^{(2)}$. It can be read off from the third line of \eqref{eq:ham}
and gives
\begin{equation}
  \label{dd}
  \cE_V[\psi] = \tr\bigl((-\tfrac12\Delta+V)\gamma_\psi\bigr)
  +\int_{\rz^3}\rd\gx\int_{\rz^3}\rd \gy\,
  \frac{\rho_\psi^{(2)}(\gx,\gy)}{|\gx-\gy|}.
\end{equation} 

\enlargethispage{1.7\baselineskip}%%++

\section{The Hohenberg--Kohn theorem \label{sHK}}
As mentioned above, Hohenberg's and Kohn's original argument merely supports
the existence of the functional $\cE^\mathrm{HK}$. The underlying
mathematical question raised by the argument is indeed intriguing and
has been the subject of recent investigation. We will, however, not
expand on these results and refer to Garrigue
\cite{Garrigue2018,Garrigue2019,Garrigue2020}. Instead, we will
present a variational construction of the Hohenberg--Kohn functional.
An early version is due to Percus \cite{Percus1978} who outlined the
idea in the context of vanishing electron-electron interaction. It is
known as \textit{Levy--Lieb constraint search} going back to Levy
\cite{Levy1979} and Lieb \cite{Lieb1983}. The variational argument
does not only show the existence of such a functional but also
indicates a way on how to approximate the Hohenberg--Kohn
functional. Again we will present it in the context of
nonrelativistic quantum mechanics with Coulomb interactions among the
electrons.

\begin{definition}
  For $N\in\rz_+$ we set\footnote{Here and below $\int_{\rz^3}\rho$
    abbreviates $\int_{\rz^3}\rd\gx\,\rho(\gx)$, and similarly for
    $\int_{\rz^3}\overline{\rho}\sigma$ and the like.}
  \begin{align}
    %\label{erlaubt}
    \cA:=&\bigl\{\rho\bigm|\rho\geq0,\ \sqrt{\rho}\in H^1(\rz^3)\bigr\},
           \nonumber\\
    %\label{erlaubtkleinerN}
    \cA_N:=&\biggl\{\rho\in\cA\biggm|\int_{\rz^3}\rho\leq N\biggr\},\nonumber\\
    %\label{erlaubtN}
    \cA_{\partial N}:= &\biggl\{\rho\in\cA\biggm|\int_{\rz^3}\rho=N\biggr\}.
  \end{align}
  If $N\in\mathbb{N}$ we call $\cA_{\partial N}$ the set of $N$-particle
  densities.
  We call the functional
  \begin{align}
    \label{HK}
       &\cE^{\mathrm{HK}}:\cA_{\partial N}\to \rz_+\cup\{\infty\},\nonumber\\
      &\rho\mapsto \inf\Biggl\{\cE_V[\psi]-\!\int_{\rz^3}\!\!\rd \gx\,
      V(\gx)\rho(\gx)\Biggm|\psi\in\gH_N\cap H^1(\Gamma^N),\; \|\psi\|=1,\;
      \rho_\psi=\rho\Biggr\}
   \end{align}
  the Hohenberg--Kohn functional.
\end{definition}
In \eqref{HK} we use the standard convention that the infimum over the
empty set is $\infty$, which in the case at hand occurs, if there is
no pure $N$-electron state $\psi\in\gH_N$ of finite kinetic energy
such that its one-particle density $\rho_\psi$ equals the given
$\rho$. Such densities are called non-$N$-representable densities. In
other words, nonrepresentable densities do not contribute to the
infimum.

The following immediate observation holds (Percus \cite{Percus1978},
Levy \cite{Levy1979}, and Lieb \cite{Lieb1983}).
\begin{theorem}
  We have
  \begin{equation}
    \inf\bigl\{\sigma(H_{V,N})\bigr\}
    = \inf\biggl\{\cE^\mathrm{HK}(\rho)+ \int_{\rz^3}\rd \gx\,
  V(\gx)\rho(\gx)\biggm|\rho\in\cA_{\partial N}\biggr\}.
  \end{equation}
  Furthermore, if $H_{V,N}$ has a
  ground state, a minimizer of the right side exists and is the one-particle
  density of a ground state.
\end{theorem}
In other words, this observation justifies the functional's name:
$\cE^\mathrm{HK}$ is the universal, i.e., independent of the external
potential $V$, functional of the reduced one-particle density matrix
which, when $\int_{\rz^3}\rd \gx\, V(\gx)\rho(\gx)$ is added, yields
the exact quantum ground state energy as minimal value and an exact
quantum one-particle ground state density as minimizer. In other
words, it is the functional whose existence Hohenberg and Kohn showed.

This construction has various more or less immediate extensions:
\begin{description}
\item\textbf{Spin-dependent potentials (Barth and Hedin
    \cite{BarthHedin1972}):} The functional depends on the diagonal of
  $\gamma$, the space-spin density, instead of $\rho$, the space
  density, only.
\item\textbf{Magnetic fields and Spin (Rajagopal and Callaway
    \cite{RajagopalCallaway1973}):} The functional depends on the
  electric current which, in state $\psi$, is
  $\gj_\psi(x):= (a_x\psi,p_\gx a_x\psi)$.
\item\textbf{$k$-body potentials (M\"uller et
    al.\ \cite{MullerSiedentop1981}):} The functional depends on the
  $k$-body density.
\item\textbf{General one-body potentials (Gilbert
    \cite{Gilbert1975}):} The functional depends on the reduced
  one-particle density matrix $\gamma$.
\end{description}

As mentioned above, Hohenberg's and Kohn's original argument for the
existence of $\cE^\mathrm{HK}$ is nonconstructive. In contrast,
although the construction \eqref{HK} appears difficult to be carried
through completely, it nevertheless offers an approximation scheme. An
example of such a scheme actually predates the work of Hohenberg and
Kohn. It is based on an idea of Macke \cite{Macke1955} carried through
by March and Young \cite{MarchYoung1958} for one-dimensional fermions
(see also Percus \cite{Percus1978}). We will outline the idea
suppressing the spin-dependence and interaction to exhibit it more
clearly: Given any density $\rho$ on $\rz$ with mass $N$ and with
${\sqrt\rho}\,'$ square integrable, we define a Slater determinant
$\psi$ with orbitals
\begin{align}\label{Macke}
    &\phi_n(x):= \sqrt{Y'(x)}\exp\bigl(\ri 2\pi (n-a) Y(x)\bigr),\
    n=1,\ldots,N,\ a\in\rz,\nonumber\\
    &Y:\rz\to(0,1),\ Y(x):= \int_{-\infty}^x\rd t\,\frac{\rho(t)}{N},
\end{align}
also known as Macke orbitals. We immediately read off that these are
orthonormal and their density is
\begin{equation}
  \label{eq:m1}
  \rho_\psi(x) =\sum_{n=1}|\phi_n(x)|^2=\rho(x).
\end{equation}
Moreover, a small calculation and optimization in the parameter $a$ shows
\begin{equation}
  \label{eq:hkma}
  \cE^\mathrm{HK}(\rho)\leq\frac12\int_\rz\rd\gx\,\biggl[\sqrt{\rho}\,'(\gx)^2
    +\frac{\pi^2}{3}\biggl(1-\frac1{N^2}\biggr)\rho(\gx)^3\biggr]
\end{equation}
which --- apart from the $-N^{-2}$ --- is the one-dimensional
Thomas--Fermi--Weizs\"acker functional of the kinetic energy (March and
Young \cite{MarchYoung1958}).

M\"uller \cite{Muller1979} found a generalization of \eqref{Macke} to
all dimensions $d$ such that the orbitals fulfill \eqref{eq:m1} for
given $d$-dimensional density $\rho(\gx)$, namely
\begin{equation}
  \label{Macked}
  \phi_\nu(\gx):= \sqrt{\bigl|\det\bigl(J(\gx)\bigr)\bigr|}
  \exp\bigl(\ri 2\pi (n_\nu-a)Y(\gx)\bigr)
 \end{equation}
for $n_\nu\in \gz^d$ with $\nu=1,\ldots,n$, and $a\in\rz^d$
with
\begin{align}
  \label{macketrafo}
  &Y:\rz^d\to(0,1)^d,\nonumber\\
  &\gx\mapsto
    \begin{pmatrix}
      \displaystyle{\frac{\int_{-\infty}^{x_1}\rd t_1\,\rho(t_1,x_2,\ldots,x_d)}
      { \int_{-\infty}^\infty\rd t_1\,\rho(t_1,x_2,\ldots,x_d)}}\\
      \vdots\\
      \displaystyle{\frac{\int_{-\infty}^\infty\rd
        t_1\cdots \int_{-\infty}^\infty\rd t_{d-1}\int_{-\infty}^{x_d}\rd
        t_d\,\rho(t_1,\ldots,t_d)}{\int_{-\infty}^\infty\rd
        t_1\cdots \int_{-\infty}^\infty\rd t_d\,\rho(t_1,\ldots,t_d)}}
    \end{pmatrix}
\end{align}
where $J$ is the Jacobian of $Y$. These orbitals are also
orthonormal. The Jacobian is the determinant of a tridiagonal matrix,
i.e., $J$ is just the product of the diagonal. It is a telescopic
product yielding $\rho(\gx)/N$. However the corresponding Slater
determinant reproduces merely the Weizs\"acker part of the kinetic
energy but not the Thomas--Fermi part.

Note that March and Young \cite{MarchYoung1958} proposed a choice for
the orbitals \eqref{Macked} with a different $Y$. They postulated
properties of the transform of $\rz^d$ to $(0,1)^d$ which, however,
lead to a contradiction. One may, however, salvage the argument in
three dimensions by a slight modification which is suitable for spherical
potential and, in this way, obtain the Hellmann-Weizs\"acker functional
\cite{Hellmann1936}
\begin{align}
  \label{Hellmann}
  \cE^\mathrm{HW}\bigl(\arvec{\rho}\bigr)
  &:= \frac12\sum_{l=0}^\infty \int_0^\infty\rd r\left(\sqrt{\rho_l}\,'(r)^2
    +\frac{l(l+1)}{r^2}\rho_l(r)
    +\frac{\pi^2}{3}\frac{\rho_l^3(r)}{\bigl(2(2l+1)\bigr)^2}\right)\nonumber\\
  &\mathrel{\phantom{=}}\mbox{}
    +\sum_{l=0}^\infty \int_0^\infty\rd r\, V(r)\rho_l(r)\nonumber\\
  &\mathrel{\phantom{=}}\mbox{}
  + \frac12\sum_{l,l'=0}^\infty \int_0^\infty\rd r \int_0^\infty\rd r'
  \,\frac{\rho_l(r)\rho_{l'}(r')}{\max\{r,r'\}}
\end{align}
with ${\arvec{\rho}=(\rho_0,\rho_1,\ldots)}$ instead of the
Thomas--Fermi--Weizs\"acker functional as an upper bound. We assume that
${\rho_0\geq0}$, ${\rho_1\geq0}$, \ldots,
${\sum_{l=0}^\infty\int_0^\infty\rd r \sqrt{\rho_l}'(r)^2<\infty}$ and
${\sum_{l=0}^\infty\int_0^\infty\rho_l\leq N}$, the electron number. The
function $\rho_l$ with ${l\in\mathbb{N}_0}$ may be interpreted as the radial
densities of electrons in angular momentum channel $l$.

Although infinitely many $l$ are allowed, it turns out that the
minimizer has only finitely many. Using Hardy's inequality, we
obtain a lower bound with the gradient term dropped and $l(l+1)$
replaced by $(l+\tfrac12)^2$. The resulting function would give a
positive contribution for high angular momenta
\cite{SiedentopWeikard1986O}. See also a bound on the total charge
obtained by the same argument yielding \eqref{RafaelW} (Benguria et al.\ 
\cite{Benguriaetal1992}).

\enlargethispage{0.9\baselineskip}%%++

The following modification of \eqref{Macke} --- additionally with spin
included --- shows that $\cE^\mathrm{HW}(\arvec\rho)$ is an upper bound
on the ground-state energy $\inf\{\sigma(H_{V,N})\}$ of the $N$-electron
system:
\begin{equation}
  \label{Mackel}
  \phi_{n,l,m,s}(x)  := \frac{\sqrt{Y_l'(|\gx|)}}{|x|}\exp\bigl(\ri 2\pi
    (n_l-a_l) Y_l(|x|)\bigr)Y_{l,m}(\gx/|\gx|)\delta_{s,\sigma}
\end{equation}
for $n=1,\ldots,N_{l,m,s}\in\mathbb{N}_0$, $m=-l,\ldots,l$, $s=1,2$,
$\sum_{l,m,s}N_{l,m,s}=N$, and $a_l\in\rz$ with
\begin{equation}
    \label{macketraforadial}
    Y_l:(0,\infty)\to(0,1)^d,\quad
    r\mapsto\frac{\int_0^r\rd t\,\rho_l(t)}
    {\int_0^\infty\rd t\,\rho_l(t)}
\end{equation}
for all $l\in\mathbb{N}_0$ with $\int_0^\infty\rho_l>0$ (see
Ladanyi \cite{Ladanyi1958} and \cite{Siedentop1981}).

\section{Some results on concrete density functionals}
In the following we will spotlight some of the basic density
functionals and density matrix functionals. To simplify the notation we
concentrate on the atomic case although there are generalizations to
molecules and other more general external potentials.

\subsection{Thomas--Fermi theory\label{tftheorie}}

\subsubsection{Definition and basic properties}

The Thomas--Fermi functional (Lenz \cite{Lenz1932}) of an atom with
nuclear charge $Z$ is
\begin{equation}
  \label{TFF}
  \cE_Z^{\mathrm{TF}}(\rho)
  :=\int_{\rz^3}\rd\gx\left(\frac35\gtf\rho(\gx)^{\tfrac53}
    -\frac Z{|\gx|}\rho(\gx)\right)
  +\underbrace{\frac12\int_{\rz^3}\rd\gx\int_{\rz^3}\rd\gy\,
    \frac{\rho(\gx)\rho(\gy)}{|\gx-\gy|}}_{=:D[\rho]}
\end{equation}
where $\gtf$ is a positive constant, namely
$(\hbar^2/2m)(6\pi^2/q)^{2/3}$ which for electrons ($q=2$) is
$(3\pi^2)^{2/3}/2$ in Hartree units ($\hbar=m=1$).

We will collect some (mostly) known results. We refer to Lieb and
Simon \cite{LiebSimon1977}, Simon \cite[Section 9]{Simon1979}, Lieb
\cite{Lieb1981}, and Lieb and Loss \cite{LiebLoss1996} for other
reviews and references.

As usual we write $\cS(\rz^3)$ for the Schwartz space of fast decaying
functions in $C^\infty(\rz^3)$ and $\cS'(\rz^3)$ for its dual, the
tempered distributions. (See, e.g., Lieb and Loss \cite{LiebLoss1996}
for more details.)

\begin{definition}
 We write
 \begin{align}
   \cC&:=\biggl\{\rho\in\cS'(\rz^3)\biggm|
        \int_{\rz^3}\rd\gxi \,
        \bigl|\cF(\rho)(\gxi)\bigr|^2/|\gxi|^2
        <\infty\biggr\}, \nonumber\\ %\label{4.2}\\\label{4.3}
   \cI&:=\Bigl\{\rho\in L^{\frac53}(\rz^3)\Bigm|
        \rho\geq0,\ D[\rho]<\infty\Bigr\},\nonumber\\
   \cI_N&:=\biggl\{\rho\in\cI\biggm|
     \int_{\rz^3}\rho\leq N\biggr\},\nonumber\\ %\label{4.4}\\\label{4.5}
   \cI_{\partial N}&:=\biggl\{\rho\in\cI_N\biggm|\int_{\rz^3}\rho=N\biggr\}.
 \end{align}
 for the set of all tempered distributions $\rho$ with finite
 electron-electron interaction TF-energy, those positive distributions
 which have also finite kinetic energy, and those with corresponding
 constraints on their masses.
\end{definition}
Note that $\cC$ is a Hilbert space with the scalar product
\begin{align}
  \label{D}
D(\rho,\sigma)
&:=\frac12 \int_{\rz^3}\rd\gx\int_{\rz^3}\rd\gy\,
   \frac{\overline{\rho(\gx)}\sigma(\gy)}{|\gx-\gy|}\nonumber\\
&= 2\pi\int_{\rz^3}\rd\gxi\,
  \frac{\overline{\cF(\rho)(\gxi)}\cF(\sigma)(\gxi)}
  {|\gxi|^2},
\end{align}
the sesquilinear form associated with the positive quadratic form
$D$. The fact that for any $\rho\in\cC$ one has $D[\rho]>0$ unless
$\rho=0$, which is immediate from its Fourier representation, is known
as Onsager's inequality.

One of the key observations which elevates Thomas--Fermi theory to an
important tool in the analysis of many electron systems is the
Lieb--Thirring inequality (Lieb and Thirring
\cite{LiebThirring1975,LiebThirring1976}).
\begin{theorem}
  \label{satzlt}
  There exists a positive constant $\gamma_\mathrm{LT}$ such that for
  all one-particle density matrices $\gamma$ with finite kinetic
  energy, i.e., for all $\gamma\in \gS^1(L^2(\Gamma))$ with
  $0\leq\gamma\leq1$ and $\Delta\gamma\in\gS^1(L^2(\Gamma))$
  \begin{equation}
    \label{lt}
    \tr\bigl(-\tfrac12\Delta\gamma\bigr)\geq
    \frac35\gamma_\mathrm{LT} \int_{\rz^3}\rd \gx\, \rho_\gamma(\gx)^\frac53.
  \end{equation}
\end{theorem}
A dual formulation is
\begin{theorem}
  There is a constant $L>0$ such that for all $\varphi\in L^\frac52(\rz^3)$
  \begin{equation}
    \label{lt2}
    \tr\bigl(-\tfrac12\Delta-\varphi\bigr)_-
    \geq -L\int_{\rz^3}\rd \gx\, \varphi_+(\gx)^\frac52.
  \end{equation}
\end{theorem}
(We indicate the positive part of a function or operator by an index
$+$, and by an index $-$ the negative part [which we pick negative,
i.e., $f=f_++f_-$].)

It is a longstanding conjecture, called the Lieb--Thirring conjecture,
that the optimal constant is given by semi-classical phase space
counting, e.g., $\gamma_\mathrm{LT}=\gtf$. (Note that the conjectured optimal
constant depends on the number of spin states $q$ per electron, namely
$\gamma= (6\pi^2/q)^\frac23/2$.) It is known, that $\gamma$ cannot be
larger than the classical one because of asymptotic results for the
sum of eigenvalues. Although over the years there have been
considerable efforts in proving the Lieb--Thirring conjecture, they
merely yielded improvements of the lower bound on $\gamma$ but not the
conjectured value. On the other hand there have been also considerable
efforts in disproving the conjecture. However, these were not
successful either.

We will not prove the Lieb--Thirring inequality and their
extensions. We refer instead to Nam's \cite{Nam2020} review in these
proceedings and make freely use of it. See also Benguria and Loewe
\cite{BenguriaLoewe2013}.

\begin{theorem}
  The Thomas--Fermi functional $\cE^\mathrm{TF}_Z$ is well defined on
  $\cI$ and bounded from below.
\end{theorem}
\begin{proof}
  Obviously the kinetic energy and the electron-electron energy are
  well defined and finite for every $\rho\in\cI$. To show that this
  holds also for the nuclear potential we decompose the Coulomb
  potential $1/|\gx|$: We write $\sigma:= \delta(R-|\gx|)/(4\pi R^2)$
  for a unit charge smeared out homogeneously on a sphere of radius $R$
  centered at the origin. Obviously $\sigma\in \cC$, since
  \begin{equation}
    D[\sigma]=\frac12\int_0^\infty\rd r \,
    \frac{4\pi r^2}{4\pi R^2}\int_0^\infty\rd s\,
    \frac{4\pi s^2}{4\pi R^2}
    \frac{\delta(R-r)\delta(R-s)}{\max\{r,s\}}=\frac{1}{2R}.
  \end{equation}
  We set
\begin{align}
    \label{vl}
    V_l&:\sigma*|\cdot|^{-1}=
    \begin{cases}
      1/|\gx| &\text{for } |\gx|\geq R,\\
      1/R& \text{for }|\gx|<R,
    \end{cases}\nonumber\\
     V_k&:=\frac{ 1}{|\cdot|}-V_l
\end{align}
yielding a decomposition of the Coulomb potential into its long range
regular part and its short range singular part.  Thus, using
H\"older's inequality on the short-range part and the Schwarz
inequality on the long-range part we get
  \begin{align}
    \int_{\rz^3}\rd\gx \,\frac{\rho(\gx)}{|\gx|}
    &\leq \int_{|\gx|<R}\rd\gx\,
    \frac{\rho(\gx)}{|\gx|} + 2D(\rho,\sigma)\nonumber \\
   & \leq {\left(\int_{|\gx|<R}\frac{\rd\gx}{|\gx|^{\frac52}}\right)}^{\frac25}
    \|\rho\|_{\frac53} +2D(\rho,\sigma)\nonumber\\
    &\leq
    (8\pi)^{\frac25}R^{\frac15}\|\rho\|_{\frac53} + 2\sqrt{D[\sigma]D[\rho]}\nonumber\\
    &\leq (8\pi)^{\frac25}R^{\frac15}\|\rho\|_{\frac53} + 2\sqrt{D[\rho]/(2R)}.
  \end{align}
  Thus, also the nuclear-electron energy is finite. Using this bound
  in the functional and minimizing in $\|\rho\|_\frac53$ and $D[\rho]$
  gives
\begin{align}
  \label{untereSchranke}
  \cE^\mathrm{TF}_Z(\rho)
  &\geq \frac35\gtf\|\rho\|_{\frac53}^{\frac53}
  -(8\pi)^{\frac25}R^{\frac15}Z\|\rho\|_{\frac53} + D[\rho]-
    Z\sqrt{2D[\rho]/R}\nonumber\\
  &\geq -\frac{32\sqrt2}{15\pi}Z^{\frac52}R^{\frac12}-
  \frac{Z^2}{2R}= -\frac{2^{\frac83}3^{\frac13}}{(5\pi)^{\frac23}}Z^{\frac73}
  \approx -1.46 Z^{\frac73}
\end{align}
after picking $R=(15\pi)^{\frac23}/(2^{\frac{11}{3}}Z^{\frac13})$.
\end{proof}
\begin{lemma}
  \label{tfstrengkonvex}
  The Thomas--Fermi functional and its restrictions to $\cI_N$ and
  $\cI_{\partial N}$ are strictly convex, i.e., the sets $\cI$,
  $\cI_N$, and $\cI_{\partial N}$ are convex sets and for every $t\in(0,1)$
  and $\rho,\tau\in\cI$ we have
  \begin{equation}
  t\cE^\mathrm{TF}_Z(\rho)+(1-t)\cE^\mathrm{TF}_Z(\tau)
  \geq \cE^\mathrm{TF}_Z\bigl(t\rho+(1-t)\tau\bigr)
  \end{equation}
  with equality if and only if $\rho=\tau$.
\end{lemma}
\begin{proof}
  The convexity of the above three sets is obvious and so is the
  strict convexity of the kinetic energy. The nuclear potential is
  linear. It remains the electron-electron interaction:
  \begin{align}
    D[t\rho+(1-t)\tau]
    &= t^2D[\rho]+(1-t)^2D[\tau] +2t(1-t)D(\rho,\tau)\nonumber\\
    &\leq t^2D[\rho]+(1-t)^2D[\tau] +2t(1-t)\sqrt{D[\rho]D[\tau]}\nonumber\\
    &= \Bigl(t\sqrt{D[\rho]}+(1-t)\sqrt{D[\tau]}\Bigr)^2 \nonumber\\
    &\leq  tD[\rho]+(1-t)D[\tau].
  \end{align}
where the first inequality is true because of the Schwarz inequality
and the last because of the convexity of the square.  
\end{proof}
We define
\begin{equation}
  \label{E}
  E^\mathrm{TF}(Z,N):= \inf\bigl\{\cE_Z^\mathrm{TF}(\cI_N)\bigr\},\
  E^\mathrm{TF}(Z):=\inf\bigl\{\cE_Z^\mathrm{TF}(\cI)\bigr\}.
\end{equation}
\begin{lemma}
  For fixed $Z\in\rz_+$ the function
  $E^\mathrm{TF}(Z,\cdot):\rz^+\to\rz$ is monotone decreasing and
  convex.
\end{lemma}
\begin{proof}
  The monotony is obvious, since $N\leq N'$ implies $\cI_N\subset\cI_{N'}$.

  To prove the convexity, we pick $N_1,N_2,\alpha_1,\alpha_2\in\rz_+$
  so that $\alpha_1+\alpha_2=1$. Then
  \begin{align}
    &\mathrel{\phantom{=}}E^{\mathrm{TF}}(Z,\alpha_1N_1+\alpha_2N_2)\nonumber\\
    &=\inf\biggl\{\cE^\mathrm{TF}_Z(\rho)\biggm|
      \rho\in \cI,\ \int_{\rz^3}\rho\leq \alpha_1N_1+\alpha_2N_2\biggr\}\nonumber\\
    &\leq \inf\biggl\{\cE^\mathrm{TF}_Z(\rho+\sigma)\biggm|
      \rho,\sigma\in \cI,\ \int_{\rz^3}\rho\leq \alpha_1N_1,\
      \int_{\rz^3}\sigma\leq\alpha_2N_2\biggr\}\nonumber\\
    &= \inf\biggl\{\cE^\mathrm{TF}_Z(\alpha_1\rho+\alpha_2\sigma)\biggm|
      \rho,\sigma\in \cI,\ \int_{\rz^3}\rho\leq N_1,\
      \int_{\rz^3}\sigma\leq N_2\biggr\}\nonumber\\
    &\leq 
      \inf\biggl\{\alpha_1\cE^\mathrm{TF}_Z(\rho)+\alpha_2\cE_Z^\mathrm{TF}(\sigma)
      \biggm|\rho,\sigma\in \cI,\ \int_{\rz^3}\rho\leq N_1,\
      \int_{\rz^3}\sigma\leq N_2\biggr\}\nonumber
      \displaybreak[0]%%++
      \\
    &\leq \inf\biggl\{\alpha_1\cE^\mathrm{TF}_Z(\rho)\biggm|\rho\in \cI,\
      \int_{\rz^3}\rho\leq N_1\biggr\}\nonumber\\
    &\mathrel{\phantom{=}}\mbox{}
      + \inf\biggl\{\alpha_2\cE_Z^\mathrm{TF}(\sigma)\biggm|
      \sigma\in \cI,\ \int_{\rz^3}\sigma\leq N_2\biggr\}\nonumber\\
    &= \alpha_1 E^{\mathrm{TF}}(Z,N_1)+\alpha_2E^{\mathrm{TF}}(Z,N_2)
  \end{align}
  which is the desired inequality proving convexity.
\end{proof}

\begin{theorem}
  The Thomas--Fermi functional and its restriction to $\cI_N$ have a
  unique minimizer.
\end{theorem}
\begin{proof}
  The uniqueness follows from the strict convexity of
  $\cE_Z^\mathrm{TF}$. Suppose that $\rho$ and $\tau$ are two
  minimizers. Then, we are led to a contradiction
  \begin{equation}
  \cE^\mathrm{TF}_Z\bigl(\tfrac12(\rho+\tau)\bigr)
  < \frac12\bigl(\cE^\mathrm{TF}_Z(\rho)+\cE^\mathrm{TF}_Z(\tau)\bigr)
  =\inf\bigl\{ \cE^\mathrm{TF}_Z(\cI)\bigr\}    
\end{equation}
  unless $\rho=\tau$ showing there is at most one minimizer in $\cI$.
  
  We now turn to the existence of a minimizer. We begin by noting that
  because of \eqref{untereSchranke} the functional $\cE^\mathrm{TF}_Z$
  is bounded from below on $\cI$, i.e.,
  $\inf\cE_Z^\mathrm{TF}(\cI)>-\infty$. Assume that $\rho_n\in\cI$ is
  a minimizing sequence, i.e.,
  \begin{equation}
    %\label{mf}
    \lim_{n\to\infty}\cE_Z^\mathrm{TF}(\rho_n)
    =\inf\bigl\{\cE_Z^\mathrm{TF}(\cI)\bigr\}
  \end{equation}
  which, of course, is also true for any subsequence of
  $\rho_n$. Again by \eqref{untereSchranke} we see that both the
  $L^\frac53$-norm and the Coulomb-norm of $\rho_n$ are bounded, since
  otherwise there would be a subsequence of $\rho_n$ which is
  minimizing and also drive the Thomas--Fermi functional to $+\infty$
  which is certainly not true, since $0\in\cI$ and
  $\cE^\mathrm{TF}_Z(0)=0<\infty$.

  Now, since $\rho_n$ is bounded in $L^\frac53$-norm and $L^\frac53$
  is reflexive, the Banach--Alaoglu theorem gives a minimizing
  subsequence --- which we call in abuse of notation again $\rho_n$ ---
  which converges weakly in $L^\frac53$ to some $\rho\in
  L^\frac53(\rz^3)$.

  Now, this sequence is also bounded in the Coulomb norm
  \begin{equation}
    \label{Coulombnorm}
    \|\rho\|_\cC:= \sqrt{D[\rho]}
  \end{equation}
  associated
  with the scalar product $D$ in $\cC$. Since $\cC$ is a Hilbert space
  and therefore reflexive, we can pick again an appropriate
  subsequence of this subsequence, again denoted by $\rho_n$, which
  also converges weakly in the Coulomb scalar product to some
  $\tilde\rho\in \cC.$ But actually, those two limits are equal: Both
  convergences imply convergence as a tempered
  distribution. Therefore, we have for any $f\in\cS(\rz^3)$
  \begin{align}
    \int_{\rz^3}\rho f
    &=\lim_{n\to\infty}\int_{\rz^3}\rho_n f
    = \lim_{n\to\infty} 2D\bigl(\rho_n,-\Delta f/(4\pi)\bigr)\nonumber\\
    &= 2D\bigl(\tilde\rho,-\Delta f/(4\pi)\bigr)=\int_{\rz^3}\tilde\rho f. 
\end{align}
Since this holds for all $f\in\cS(\rz^3)$, we have
$\rho=\tilde\rho\in L^\frac53(\rz^3)\cap \cC$. In other words, the
subsequence $\rho_n$ is minimizing and converges weakly in
$L^\frac53(\rz^3)$ and $\cC$ to $\rho$.

Now we note that $\rho\geq 0$ almost everywhere. Suppose that there is
a set $\cN\subset\rz^3$ with $\int_\cN>0$ and $\rho(\cN) <0$. Then there
exists a radius $R\in\rz_+$ such that $\int\chi_R\rho<0$ where $\chi_R$
denotes the characteristic function of the ball of radius $R$ centered
at the origin intersected with $\cN$. This function is, of course, in
any $L^p$ space, in particular in $L^\frac52$, the dual space of
$L^\frac53$. Then we are led to the following contradiction
  \begin{equation}
    0>\int_{\rz^3}\chi_R\rho=\lim_{n\to\infty}\int_{\rz^3}\chi_R\rho_n\geq0.
  \end{equation}
  Thus, $\rho$ is nonnegative and with the above result that
  $\rho\in L^\frac53(\rz^3)\cap\cC$ we conclude that $\rho\in\cI$.

  We claim that the function $\rho$, which we just constructed, is a
  minimizer. We prove this term by term:\newline\noindent%
  \textit{The kinetic energy:} Since $\rho\in L^\frac53$, we have $\rho^\frac23\in L^\frac52$. Thus
  \begin{equation}
    \int\rho^\frac53 = \int\rho^\frac23 \rho
    =\lim_{n\to\infty}\int\rho^\frac23\rho_n
    \leq \liminf_{n\to\infty}\left(\int\rho^\frac53\right)^\frac35
    \left(\int\rho_n^\frac53\right)^\frac25
  \end{equation}
  by weak convergence in $L^\frac53$ and H\"older's inequality. Therefore
  \begin{equation}
  \int \rho^\frac53\leq \liminf_{n\to\infty}\int\rho_n^\frac53.  
  \end{equation}  
 \textit{The nuclear potential}: Since $V_k\in L^\frac52$ we have
 because of weak convergence in $L^\frac53$
 \begin{equation}
  \int V_k\rho=\lim_{n\to\infty}\int V_k\rho_n 
 \end{equation}  
  and since $\sigma\in \cC$ we have because of weak convergence in $\cC$
  \begin{equation}
    \label{zerlegung}
    \int V_l\rho = 2D(\sigma,\rho)=2\lim_{n\to\infty}D(\sigma,\rho_n)
    = \lim_{n\to\infty}V_l\rho_n.
  \end{equation}
\textit{The electron-electron-potential:} Since $\rho\in\cC$ we have
  \begin{equation}
    D[\rho]=D(\rho,\rho)=\lim_{n\to\infty}D(\rho,\rho_n)
    \leq \liminf_{n\to\infty}\sqrt{D[\rho]D[\rho_n]}
  \end{equation}
  by Schwarz's inequality and thus
  \begin{equation}
    D[\rho]\leq\liminf_{n\to\infty}D[\rho_n].
  \end{equation}
Putting all terms together yields
  \begin{equation}
    \cE^\mathrm{TF}_Z(\rho)\leq \liminf_{n\to\infty}\cE_Z^\mathrm{TF}(\rho_n).
  \end{equation}
  In other words, $\rho$ is the unique minimizer on $\cI$.
  
  The argument for $\cI_N$ is almost identical and therefore skipped
  here.
\end{proof}

We are now interested in some properties of the minimizer and the
minimal energy.
\begin{theorem}
  \label{minimierereigenschaften}
  The minimizer $\rho_Z$ of $\cE_Z^\mathrm{TF}$ on $\cI$ is
  spherically symmetric, decreasing and convex in the radial variable
  $|\gx|$, $\int_{\rz^3}\rho_Z=Z$, and it fulfills the Thomas--Fermi
  equation
  \begin{equation}
    \label{tf}
    \gtf\rho_Z^\frac23=\fitf:=Z|\cdot|^{-1}- \rho_Z*|\cdot|^{-1}
  \end{equation}
  almost everywhere in $\rz^3$ and the scaling relations
  \begin{equation}
    \rho_Z(\gx)=Z^2\rho_1(Z^\frac13\gx),\ E^\mathrm{TF}(Z)
    = E^\mathrm{TF}(1)Z^{7/3}.
  \end{equation}
\end{theorem}

Before turning to the proof, we would like to comment on this result:
\begin{itemize}
\item Since there is a minimizer of the Thomas--Fermi functional in
  $\cI$ this ensures the existence of a solution of the Thomas--Fermi
  equation in $\cI$.
  \item Numerically $E_\mathrm{TF}(1)= -0.7687\ [\mathrm{Ha}]$.

    On the one hand this compares with $-0.5\ [\mathrm{Ha}]$ for the ground
    state energy of the Schr\"odinger Hamiltonian of hydrogen. That it
    reproduces the order of magnitude of a one-electron system --
    despite counting also the self-energy -- is astonishing, since, as
    we will see, Thomas--Fermi theory, traditionally called a
    statistical theory of atoms, becomes correct for large electron
    numbers and is not meant for small electron numbers. However, the
    fact that it is lower than the quantum energy and therefore lower
    when dropping the electron-electron interaction, is a particular
    case of a long standing conjecture of Lieb and Thirring (see,
    e.g., Nam \cite{Nam2020} in these proceedings).

    On the other hand the value $-0.7687 Z^\frac73 [\mathrm{Ha}]$
    compares with $-1.46 Z^\frac73\ [\mathrm{Ha}]$ in
    \eqref{untereSchranke} showing that that estimate does not only
    produce the correct power law but also a numerical factor of the
    right order of magnitude.
  \item Since the unrestricted minimizer has charge $Z$, there are no
    negative ions in Thomas--Fermi theory. This might look strange at
    first sight, since negative ions are known to exist. However,
    doubly or higher negatively charged ions are unknown (Massey
    \cite{Massey1976,Massey1979}), i.e., the prediction is off by one
    only. Moreover, Thomas--Fermi theory plays an essential role in
    bounding the excess charge of atoms in more elaborate models like
    Hartree--Fock theory. There it is used to successively screen out
    the inner electrons (Solovej \cite{Solovej2003}).
  \end{itemize}

    \begin{proof}
    \textit{Spherical symmetry:} The density $\rho^\cR$ defined
    by $\rho^\cR(\gx)=\rho_Z(\cR\gx)$ is again a minimizer for any
    rotation $\cR$ about the origin. Therefore, since the minimizer is
    unique, $\rho^\cR=\rho_Z$. Since this holds for all $\cR$, the
    minimizer must be spherically symmetric.

  \textit{Upper bound on the number of electrons:} We will use
  Benguria's famous --- unfortunately unpublished --- variational
  argument: Suppose ${\int_{\rz^3}\rho_Z>Z}$. Then there is a radius
  ${R\in\rz_+}$ such that ${\int_{|\gx|\leq R}\rho_Z=Z}$. Set
  \begin{equation}
    n(\gx):=
    \begin{cases}
      \rho_Z(\gx)& \text{for\ } |\gx|\leq R\\
      0& \text{for\ } |\gx|>R
    \end{cases}
  \end{equation}
  and $\delta:= \rho_Z-n$ where $\delta$ is not vanishing almost
  everywhere. We claim that $n$ has a lower energy than the infimum
  unless $\delta$ vanishes:
  \begin{align}
    \label{benguriaus}
    \cE^\mathrm{TF}_Z(n)-\cE^\mathrm{TF}_Z(\rho_Z)
    &= -\frac35\gtf\int_{\rz^3}\rd \gx\,\delta(\gx)^\frac53\nonumber\\
    &\mathrel{\phantom{=}}\mbox{}
    +\underbrace{\int_{\rz^3}\rd\gx\,\frac{Z}{|\gx|}\delta(\gx) -2D(n,\delta)}_{=0}
    -D[\delta]  <0
  \end{align}
  where the underbraced quantity vanishes because the electric
  potential of $n$ outside the ball of radius $R$ is exactly $Z/|\gx|$
  by Newton's theorem. The last inequality in \eqref{benguriaus} is
  strict, since $\delta$ does not vanishes almost everywhere. This, of
  course, contradicts the fact that $\rho_Z$ is a minimizer. Therefore
  the supposition is absurd and
  \begin{equation}
    \label{obereSchranke}
    \int_{\rz^3}\rho_Z\leq Z.
  \end{equation}

  \textit{Thomas--Fermi equation:} We defer the lower bound on the
  charge of $\rho_Z$ and first turn to the Thomas--Fermi equation. We
  begin with the Thomas--Fermi potential $\fitf$. The electronic
  potential of the minimizer $\psi:=\rho_Z*|\cdot|^{-1}$ is H\"older
  continuous: For any nonzero $\ga\in\rz^3$ we have
  \begin{align}
    \label{holderstetig}
    \bigl|\psi(\gx+\ga)-\psi(\gx)\bigr|
    &= \left|\int_{\rz^3}\rd \gy \,\rho_Z(\gy)
      \left(\frac{1}{|\gx-\gy +\ga|}-\frac{1}{|\gx-\gy|}\right)\right|
    \nonumber\\
    &\leq \int_{\rz^3}\rd \gy\, \rho_Z(\gy)
      \frac{|\ga|}{|\gx-\gy +\ga|\,|\gx-\gy|}\nonumber\\
    &\leq
    \|\rho_Z\|_\frac53|\ga|\left(\int_{\rz^3}\rd\gy\,
      |-\gy+\ga|^{-\frac52}|\gy|^{-\frac52}\right)^\frac25\nonumber\\
    &\leq\const \|\rho_Z\|_\frac53 |\ga|^\frac15
  \end{align}
  which shows the H\"older continuity. Thus the Thomas--Fermi potential
  is continuous outside the origin where it is also subharmonic, since
  \begin{equation}
    -\frac1{4\pi}\Delta\fitf= Z\delta_0-\rho_Z.
  \end{equation}

  Pick now any positive $R$ and $ \epsilon$ and set
  \begin{equation}
    \label{test}
    \chi:= \rho_Z \sgn\bigl(\gtf\rho_Z^\frac23-\fitf\bigr).
  \end{equation}
  Note that
  $\rho_Z+\alpha \chi\in \cI$ for all $\alpha\in[-1,1]$. Moreover,
  \begin{equation}
    F(\alpha):=\cE_Z^\mathrm{TF}(\rho_Z+\alpha \chi) - \cE_Z^\mathrm{TF}(\rho_Z)
  \end{equation}
  is continuously differentiable in $\alpha$ for $\alpha\in(-1,1)$ and
  \begin{align}
    0=F'(0)&= \gtf\int_{\rz^3}\rd
    \gx\,\rho_Z(\gx)^\frac23\chi(\gx)
    -\int_{\rz^3}\rd\gx\,\frac{Z\chi(\gx)}{|\gx|}
    + 2D(\chi,\rho_Z)\nonumber\\
    &=\int_{\rz^3}\rd \gx
    \left|\gtf\rho_Z(\gx)^\frac23-\fitf(\gx)\right|\rho_Z(\gx),
  \end{align}
  since $\rho_Z$ is the minimizer. Thus,
  \begin{equation}
    \label{tfvorl}
    \gtf\rho_Z(\gx)^\frac23-\fitf(\gx)=0\ \text{for}\ \text{almost all}\ \gx\ \text{such that}\ \rho_Z(\gx)>0.
  \end{equation}
   Repeating the above argument but picking $\chi$ as the
  characteristic function of any ball $B_R(\gy)\cap \mathcal{N}$ with
  $\gy\in\rz^3$ and $R>0$ as a perturbing function, where
  $\mathcal{N}:=\{\gx\in\rz^3| \rho_Z(\gx)=0\}$, yields
  \begin{equation}
    \label{TFnull}
    0\leq F'(0)=\int_{B_R(\gy)\cap \mathcal{N}}\rd \gx\,\bigl(-\fitf(\gx)\bigr).
  \end{equation}
  (Note that we cannot conclude equality, since $\alpha$ is restricted
  to $[0,\infty)$ because of the requirement that
  $\rho_Z+\alpha \chi\in\cI$.)

  Since $\psi$ is continuous by \eqref{holderstetig} and
  $\fitf=Z/|\cdot|-\psi$, there exists a neighborhood of the origin
  which is disjoint from $\mathcal{N}$. But outside the origin $\fitf$
  is continuous. Thus $\fitf\leq0$ on $\mathcal{N}$. However, from
  \eqref{obereSchranke} we know because of spherical symmetry and
  Newton's theorem that $\fitf\geq0$ everywhere. Thus $\fitf=0$ on
  $\mathcal{N}$.

  Combining this with \eqref{tfvorl} yields
  \begin{equation}
  \gtf\rho_Z(\gx)^\frac23=\fitf(\gx)  
  \end{equation}
  almost everywhere on $\rz^3$, i.e., yields the claimed formula \eqref{tf}.
  (We note in passing --- since not needed in our proof --- that one
  can actually show that $\rho_Z>0$ everywhere. But since $\fitf$
  can only vanish if $\rho_Z$ has compact support, this implies that
  $\mathcal{N}=\emptyset$.)

    \textit{Lower bound on the number of electrons:} Now suppose
  ${\int_{\rz^3}\rho_Z <Z}$. Then, by Newton's theorem there exist ${R>0}$
  and ${\epsilon>0}$ such that for ${|\gx|>R}$ we have
  ${\fitf(\gx)\geq\epsilon/|\gx|}$. Thus, by \eqref{obereSchranke} and
  \eqref{tf}, we have
  \begin{equation}
    Z\geq\int_{\rz^3}\rho_Z =\int_{\rz^3}\bigl(\fitf(\gx)/\gtf\bigr)^{\frac32}
    =(\epsilon/\gtf)^{\frac32}\int_{|\gx|>R}\!\rd\gx/|\gx|^{\frac32}=\infty
  \end{equation}
  which is a contradiction. Thus, also $\int_{\rz^3}\rho_Z\geq Z$ and
  therefore $\int_{\rz^3}\rho_Z=Z$ as claimed.

  \textit{Decrease of the minimizer:} For showing that the minimizer
  $\rho_Z$ is decreasing in the radial variable, we remark that by the
  Thomas--Fermi equation \eqref{tf} it suffices to show that the
  radial derivative of the Thomas--Fermi potential $\fitf$ is
  negative. Thus we compute for $\gx\neq0$ and get using Newton's
  theorem
  \begin{align}
    \frac{\gx}{|\gx|}\cdot \grad \varphi_Z(\gx)
    &= \frac{\gx}{|\gx|}\cdot
    \grad_\gx {\left(\frac{Z}{|\gx|}-\int_{\rz^3} \rd\gy\,
      \frac{\rho_Z(\gy)}{\max\{|\gx|,|\gy|\}}\right)}\nonumber\\
    &=-\frac{Z}{|\gx|^2}
      +\frac{\int_{|\gy|<|\gx|}\rd\gy\,\rho_Z(\gy)}{|\gx|^2}
    \leq 0
  \end{align}
  using \eqref{obereSchranke} in the last step. This proves that
  $\rho_Z$ is decreasing in the radial variable. Differentiating again
  yields\enlargethispage{0.4\baselineskip}%%++
  \begin{align}
    \biggl(\frac{\gx}{|\gx|}\cdot \grad\biggr)^2 \varphi_Z(\gx)
    &= \frac{2}{|\gx|^3}
    {\left(Z-\int_{|\gy|<|\gx|}\rd \gy\, \rho_Z(\gy)\right)} +
    \frac{4\pi\rho_Z(\gx)}{|\gx|^2} \nonumber\\
    &\geq\frac{4\pi\rho_Z(\gx)}{|\gx|^2}\geq 0
  \end{align}
  using that ${\int_{\rz^3}\rho_Z=Z}$. This shows the convexity of
  $\fitf$ which, in turn, implies the the convexity of
  $\fitf^\frac32$ and thus the claimed convexity of $\rho_Z$.

  \textit{Scaling relations:} For any $\rho\in\cI$ a straightforward
  computation shows
  \begin{equation}
  \cE_Z^{\mathrm{TF}}\bigl(Z^2\rho(Z^\frac13\cdot)\bigr) =
  Z^\frac73\cE_1^\mathrm{TF}(\rho).  
  \end{equation}
  The scaling relations follow.
\end{proof}

Theorem \ref{minimierereigenschaften} implies that minimizers of the
Thomas--Fermi functional fulfill the Thomas--Fermi equation. The reverse is
also true:
\begin{theorem}
  If $\rho\in\cI$ fulfills \eqref{tf}, then it minimizes
  $\cE^\mathrm{TF}_Z$ on $\cI$.
\end{theorem}
We remark that this implies that the solution of the Thomas--Fermi
equation is unique in $\cI$, since there is exactly one minimizer of
the Thomas--Fermi functional in $\cI$. In particular $\rho$ is equal
to $\rho_Z$, the unique minimizer of the Thomas-Fermi functional on
$\cI$.
\begin{proof}
  Suppose that $\sigma,\rho\in\cI$ and $\rho$ fulfills \eqref{tf}. We
  set $\delta:=\sigma-\rho$. The condition that both $\rho$ and
  $\sigma$ are nonnegative implies that $\delta\geq-\rho$. We set
  \begin{equation}
    F(t):=\cE^\mathrm{TF}_Z(\rho+t\delta)-\cE^\mathrm{TF}_Z(\rho).
  \end{equation}
  In this notation, we wish to show that $F(1)\geq 0$. We compute
  \begin{align}
    \label{4.37-0}
    F(1)&=\int_0^1\rd \tau\, \dot F(\tau)\nonumber\\
    &=\int_0^1\rd \tau\left[\int_{\rz^3}\rd\gx\, \delta(\gx)
      {\left(\gtf\bigl(\rho(\gx)+\tau\delta(\gx)\bigr)^{\frac23}
      -\fitf(\gx)\right)} +2\tau D[\delta]\right]\nonumber\\
    &=\int_0^1\rd \tau\left[\gtf\int_{\rz^3}\rd\gx\, \delta(\gx)
      {\left(\bigl(\rho(\gx)+\tau\delta(\gx)\bigr)^{\frac23}
      -\rho(\gx)^{\frac23}\right)} +2\tau D[\delta]\right]\nonumber\\
    &\geq D[\delta] \geq 0
  \end{align}
  where the first inequality holds because of ${\delta\geq-\rho}$ and
  therefore the integrand of the space integral is positive regardless
  of the sign of $\delta(\gx)$. 
\end{proof}
We remark that both inequalities are strict unless $\delta$ vanishes
almost everywhere, i.e., we have
$\cE^\mathrm{TF}_Z(\sigma)>\cE^\mathrm{TF}_Z(\rho)$ unless
$\rho=\sigma$ almost everywhere.
\begin{theorem}
  For $N\leq Z$ the minimizer of the Thomas--Fermi functional on
  $\cI_N$ occurs in $\cI_{\partial N}$; whereas for $N>Z$ the
  minimizer occurs in $\cI_{\partial Z}$. In fact the Thomas--Fermi
  functional has no minimizer in $\cI_{\partial N}$ for $N>Z$.
\end{theorem}
\begin{proof}
  1. $N\leq Z$. Suppose $\rho$ minimizes $\cE^\mathrm{TF}_Z$ on
  $\cI_N$ and ${N_\mathrm{min}:=\int_{\rz^3}\rho<N}$. Then $\rho$ is
  also a minimizer of the Thomas--Fermi functional for all
  ${N'\in[N_\mathrm{min},N]}$, i.e., $E^\mathrm{TF}(Z,N')=\const$ for
  all ${N'\in[N_\mathrm{min},N]}$. We will show that this leads to a
  contradiction. Since, by Theorem \ref{minimierereigenschaften},
  $E^\mathrm{TF}(Z,N)$ is monotone decreasing and convex, we even have
  that $E^\mathrm{Z,N'}=C$ for all $N'\geq N_\mathrm{min}$, in
  particular $N'=Z$. However, for this case, we know the minimizer is
  $\rho_Z$ which is certainly different from $\rho$, since
  $\int_{\rz^3}\rho_Z>\int_{\rz^3}\rho$ by assumption; but this
  contradicts the uniqueness of the minimizer, i.e., since
  $\rho\in\cI_N$, we have $\int_{\rz^3}\rho= N$.
    
  2. $N>Z$. By Theorem \ref{minimierereigenschaften} the unique
  minimizer on $\cI$ occurs in $\cI_{Z}$. But of course, it is also a
  minimizer in $\cI_{N}$.

  Finally, to show that there is no minimizer in $\cI_{\partial N}$
  for $N>Z$, we note that
  \begin{equation}
    \label{EN=EDN}
    \inf\bigl\{\cE^\mathrm{TF}_Z(\cI_N)\bigr\}
    = \inf\bigl\{\cE^\mathrm{TF}_Z(\cI_{\partial N})\bigr\}.
  \end{equation}
  This is obvious for $N\leq Z$ and can be achieved for $N>Z$ by
  pushing the charge exceeding $Z$ to infinity and diluting it there:
  We pick a nonnegative $g\in C_0^\infty(\rz^3)$ with
  $\int_{\rz^3}g=N-Z$ and use
  $\sigma_\lambda:=\rho_Z+ \lambda^3g(\lambda x)$ with
  $\lambda>0$ as trial function. Then, obviously,
  $\sigma_\lambda\in\cI_{\partial N}$.  We claim that
  \begin{equation}
    \label{N-Z}
    \cE^\mathrm{TF}_Z(\sigma_\lambda)\to\cE^\mathrm{TF}_Z(\rho_Z)
  \end{equation}
  as $\lambda\to0$. If \eqref{N-Z} were true, then, indeed,
  \eqref{EN=EDN} would hold and the proof would be complete.

  We prove the remaining convergence \eqref{N-Z} term by term:

  \noindent%
  1. The kinetic energy: By the inverse triangular inequality we have
  \begin{equation}
    0\leq \|\sigma_\lambda\|_{\frac53}-\|\rho_Z\|_{\frac53}
    \leq\|\sigma_\lambda-\rho_Z\|_{\frac53}
    =\|\lambda^3g(\lambda\cdot)\|_{\frac53}
    = \lambda^{\frac65}\|g\|_{\frac53}\to0 
  \end{equation}
  as $\lambda\to0$.

  \noindent%
  2. The nuclear attraction: Obviously, as $\lambda\to0$
  \begin{align}
    &\mathrel{\phantom{=}}
      \int_{\rz^3}\rd\gx\,
      \frac{\rho_Z(\gx)+\lambda^3 g(\lambda\gx)}{|\gx|}
    -\int_{\rz^3}\rd\gx\,\frac{\rho_Z(\gx)}{|\gx|}\nonumber\\
    &=\int_{\rz^3}\rd\gx\,
      \frac{\rho_Z(\gx)+\lambda g(\lambda\gx)}{|\gx|}
    =\int_{\rz^3}\rd\gx\,\frac{\lambda^3 g(\lambda\gx)}{|\gx|}
    =\lambda\int_{\rz^3}\rd\gx\,\frac{g(\gx)}{|\gx|}\to 0.
  \end{align}

  \noindent%
  3. The electron-electron repulsion: By the inverse triangular
  inequality for the Coulomb norm \eqref{Coulombnorm} we have
  \begin{equation}
    0\leq \|\sigma_\lambda\|_\cC
  -\|\rho_Z\|_\cC\leq \|\sigma_\lambda-\rho_Z\|_\cC=
    \|\lambda^3g(\lambda\cdot)\|_\cC=\sqrt{\lambda}\|g\|_\cC\to0
  \end{equation}
  as $\lambda\to0$.
\end{proof}

\subsubsection{Asymptotic exactness of Thomas--Fermi theory}
The heuristic derivation of the Thomas--Fermi theory --- see, e.g.,
Gombas \cite{Gombas1949} for a textbook treatment --- may be viewed as
a semiclassical approximation with effective Planck constant
$Z^{-\frac13}$. It is therefore reasonable to guess that it describes
large --- say for simplicity --- neutral atoms correctly. We will see in
this subsection that this is indeed the case for the ground-state
energy $E^\mathrm{S}(Z):=\inf\bigl\{\sigma(H_{Z/|\gx|,Z})\bigr\}$ and the
reduced one-particle ground-state density. We begin with the energy:
\begin{theorem}
  \label{satzenergieasymptotik}
  \begin{equation}
    \label{energieasymptotik}
    E^\mathrm{S}(Z)= E^\mathrm{TF}(Z)+O(Z^\frac{25}{11})
  \end{equation}
  as $Z\to\infty$.
\end{theorem}
There are several proofs of this result. The historical first one uses
Courant's \cite[6th Chapter, \S 4]{CourantHilbert1968I}
Dirichlet-Neumann bracketing (Lieb and Simon \cite{LiebSimon1977}) and
is close to the heuristic derivation. Here we will use coherent states
(see, e.g., Thirring \cite{Thirring1980}) which were used in this
context by Lieb \cite{Lieb1981} and Thirring \cite{Thirring1981}. We
are guided by their presentations. However, we start the lower bound
differently by using a different correlation inequality.

Before embarking on the proof we introduce coherent states and review
some properties which we will use. Given ${g\in H^1(\rz^3)}$ with
${\|g\|_2=1}$, ${\vec{q},\gp\in\rz^3}$, and ${s\in\{1,2\}}$ we call the
function  
\begin{equation}
  \label{koharent}
   f_{\gp,\vec{q},s}:\;\Gamma:=\rz^3\times\{1,2\}\to \cz,\quad
   x\mapsto e^{\ri \gp\cdot\gx}g(\gx-\vec{q})\delta_{s,\sigma}
\end{equation}
a coherent state. The perhaps most well known coherent states use a
\mbox{Gaussian} as $g$, since they yield equality in the Heisenberg
uncertainty relation. Here, however, it will be practical to pick
$g$ as function of compact support. To be definite we pick it as the
ground state of a particle in a ball of radius $R$ with Dirichlet
boundary conditions continued by zero outside that ball. We pick
\begin{equation}
  \label{g}
  g_R(\gx):= R^{-\frac32} g_1(\gx/R)
\end{equation}
where
\begin{equation}
  \ g_1(\gx) =
  \begin{cases}
    \displaystyle\frac{\sin(\pi|\gx|)}{\sqrt{2\pi}|\gx|}
    &\text{for\ } |\gx|<1\\
    0&\text{for\ }|\gx|>1
  \end{cases}
\end{equation}
and  $R$ is a positive parameter that will be optimized later.

We set $z:=(\gp,\vec{q},s)\in P:=\rz^6\times\{1,2\}$,
$\int_P\rd z:=\int_{\rz^3}\dbar\gp\int_{\rz^3}\rd\vec{q}\sum_{s=1}^2$, and
$\Pi_{z}:=|f_{z}\rl f_{z}|$ where $\dbar \gp$ denotes the volume
element divided by the cube of the Planck constant $h$ (in Hartree units
$h=2\pi$).  We will pick the radius $R$ later. These states
have  easily verified interesting properties:
\begin{align}
  %\label{4.35}
  \int_P\rd z\,\Pi_z&=\1,\nonumber\\
  \label{4.36}
  0\leq\gamma(z)\leq 1\,&\mathrel{\!\!\implies}
  0\leq \int_P\rd z \,\gamma(z)\Pi_z\leq1,\\
  \label{4.37}
  \int_P\rd z\, \gamma(z)&= \tr\biggl(\int_P\rd z \,\gamma(z)\Pi_z\biggr)
\end{align}
where the integrals are understood in the weak sense, i.e., e.g.,
$\int_P\rd z\, \gamma(z)\Pi_z$ is the operator $O$ whose matrix elements
$(f,Og)$ are $\int_P\rd z\, \gamma(z)(f,\Pi_zg)$ for every
$f,g\in L^2(\Gamma,\rd x)$. (As reminder, the notation $A\leq B$ for
two selfadjoint operators means, that the domain $D(B)$ of $B$ is
included in $D(A)$ and for all $f\in D(B)$ the inequality
$(f,Af)\leq (f,Bf)$ holds.)
\begin{proof}[Theorem \ref{satzenergieasymptotik}]
  \textit{Upper bound:} We pick
  \begin{equation}
    \label{ansatzks}
    \gamma:= \int_P\rd z\,
    \theta\bigl(-\bigl[\tfrac12\gp^2-\fitf(\gq)\bigr]\bigr)\Pi_z
  \end{equation}
  which --- by \eqref{4.36} and \eqref{4.37} --- is in
  $\cD_{\partial Z}$ (see \eqref{DDN} for the notation).  We insert
  this into the Hartree--Fock variational principle \eqref{hf}. Since
  it is an upper bound on the quantum energy by \eqref{os} and the
  exchange term is negative, we get --- using our choice of $g_R$, in
  particular using that $g_R$ is the ground state eigenfunction of the
  Laplacian on the ball of radius $R$ which has eigenvalue
  $\pi^2/R^2$ ---
  \begin{align}
    \label{4.39}
    &\mathrel{\phantom{=}}
      E^\mathrm{S}(Z) \leq \cE^\mathrm{HF}_Z(\gamma)\nonumber\\
    &\leq
    \int_{\tfrac12\gp^2-\fitf(\vec{q})<0}\rd z\int_\Gamma\rd x
      \left(\frac12|\grad f_z(x)|^2
      -\frac{Z|f_z(x)|^2}{|\gx|}\right)\nonumber\\
    &\mathrel{\phantom{=}}\mbox{}
    +\frac12\int_\Gamma\rd x\int_\Gamma\rd y\,
    \frac{\int_{\tfrac12\gp^2-\fitf(\vec{q})<0}\rd z
    \int_{\tfrac12\tilde\gp^2-\fitf(\tilde{\vec{q}})<0}\rd \tilde z\,
    |f_z(x)|^2|f_{\tilde z(y)}|^2}{|\gx-\gy|}\nonumber\\
    &\leq \int_{\tfrac12\gp^2-\fitf(\vec{q})<0}\rd z
    \left(\frac{|\gp|^2}{2}+\frac{\pi^2}{2R^2}-
      |g_R|^2*\frac{Z}{|\cdot|}(\vec{q})\right)\nonumber\\
    &\mathrel{\phantom{=}}\mbox{}
    +\frac12 \int_{\rz^3}\rd\vec{q} \int_{\rz^3}\rd\tilde{\vec{q}}\,
    g_R^2*\frac{1}{|\cdot|}*g_R^2(\vec{q}-\tilde{\vec{q}}).
  \end{align}
  By Newton's theorem we have
  $|g_R|^2*|\cdot|^{-1}(\vec{q})\geq |\vec{q}|^{-1}\theta(|\vec{q}|-R)$ and
  that   $g_R^2*|\cdot|^{-1}*g_R^2(\vec{q}-\tilde{\vec{q}})
  \leq1/|\vec{q}-\tilde{\vec{q}}|$. Moreover,
  a direct computation shows%\enlargethispage{1.0\baselineskip}%%++
  \begin{equation}
    \label{rhotf}
    \rho_Z(\vec{q})
    = \int_{\tfrac12\gp^2-\fitf(\vec{q})<0}\dbar\gp\sum_{\sigma=1}^2
  \end{equation}
  and
  \begin{equation}
    \label{Ttf}
    \frac35\gtf\rho_Z(\vec{q})^{\frac53}
    = \int_{\tfrac12\gp^2-\fitf(\vec{q})<0}\dbar\gp\,
    \sum_{\sigma=1}^2\frac12|\gp|^2.
  \end{equation}
  Thus, continuing \eqref{4.39} we get
  \begin{align}
    E^S(Z)
    &\leq
      2\int_{\rz^3}\dbar\gp
      \int_{\rz^3}\rd\vec{q}
      \left(\frac{|\gp|^2}{2}-\fitf(\vec{q})\right)-D[\rho_Z]
      + \frac{\pi^2Z}{2R^2}\nonumber\\
    &\mathrel{\phantom{=}}\mbox{}
      + \int_{|\gx|<R}\rd \gx\, \frac{Z\rho_Z(\gx)}{|\gx|}\nonumber\\
    &\leq \cE^\mathrm{TF}_Z(\rho_Z)
      +\frac{\pi^2Z}{2R^2}
      -\frac{8\pi Z^{\frac52}R^{\frac12}}{\gtf^{\frac32}}
      = E^{\mathrm{TF}}(Z)+\const Z^{\frac{11}{5}}
  \end{align}
  where we use $\gtf\rho(\gx)^\frac23\leq Z/|\gx|$ and picked
  $R=Z^{-\frac35}$ which proves the upper bound.

  \textit{Lower bound:} The first step of the lower bound is to
  eliminate the correlation in favor of a mean field using a
  correlation inequality. There are several ways of doing so. We will
  use the one by Mancas et al.\  \cite[Formula (14)]{Mancasetal2004}
  \begin{equation}
    \label{Mancas}
    \sum_{1\leq m<n\leq N}\frac{1}{|\gx_m-\gx_n|} \geq 
    \sum_{n=1}^N\int_{|\gy-\gx_n|>R(\gx_n)}
    \frac{\rd\gy\,\sigma(\gy)}{|\gx_n-\gy|}
    -D[\sigma]
  \end{equation}
  where $\sigma\in\cC$ will be specified later with $\int\sigma>\frac12$
  and $R(\gx)$ is defined by
  \begin{equation}
    \label{loch}
    \int_{|\gy-\gx|<R(\gx)}\rd\gy\, \sigma(\gy)=\frac12,
  \end{equation}
  i.e., it is a ball, called the exchange hole, centered at $\gx$ with
  radius $R(\gx)$ containing exactly half an electron.

  The right side of \eqref{Mancas} can be estimated further: we have
  \begin{equation}
    \int_{|\gx-\gy|\leq R(\gx)}\rd \gy\,\frac{\sigma(\gy)}{|\gx-\gy|}
    \leq \frac12\sqrt[3]{9\pi(M\sigma)(\gx)}
  \end{equation}
  by \cite[Formula (19)]{Mancasetal2004} where $Mf$ denotes the
  maximal function of the function $f$, i.e.,
  \begin{equation}
    (Mf)(\gx):= \sup_{R>0}{\left\{
        \frac{\int_{|\gx-\gy|<R}\rd \gy\, f(\gy)}{\frac{4\pi}{3}R^3}\right\}}.
  \end{equation}
  Thus \eqref{Mancas} implies
  \begin{equation}
    \label{mancasmod}
     \sum_{1\leq m<n\leq N}\frac{1}{|\gx_m-\gx_n|} \geq
     \sum_{n=1}^N{\left(\int_{\rz^3}
       \frac{\rd\gy\,\sigma(\gy)}{|\gx_n-\gy|}
       -\frac{\sqrt[3]{9\pi}}{2}\sqrt[3]{(M\sigma)(\gx_n)}\right)}
    -D[\sigma].
  \end{equation}
  This allows us to estimate 
  \begin{equation}
    H_{-Z/|\cdot|,N}
    \geq \sum_{n=1}^N{\left(-\tfrac12\Delta_n-\varphi_\sigma(\gx_n)
      -\tfrac12\sqrt[3]{9\pi(M\sigma)(\gx_n)}\right)}
    -D[\sigma]
  \end{equation}
  with $\varphi_\sigma=Z/|\cdot|-\sigma*|\cdot|^{-1}$. Thus we get the
  lower bound
  \begin{equation}
    \label{lb}
    E^\mathrm{S}(Z)
    \geq \inf_{\gamma\in\cD_N}{\left\{
    \tr{\left[\left(-\tfrac12\Delta-\varphi_\sigma
        - \tfrac12\sqrt[3]{9\pi(M\sigma)}\right)\gamma\right]}\right\}}
    -D[\sigma].
  \end{equation}
  
  Next note that for $\xi\in H^1(\Gamma)$
  \begin{equation}
    \|\grad\xi\|^2 = \int_P\rd z\, (\xi,\grad f_z)\cdot(\grad
    f_z,\xi)= \int_{P} p^2(\xi,\Pi_z\xi) - \|\grad g_R\|^2
  \end{equation}
  and
  \begin{equation}
    \label{T}
    \int_{\Gamma}\rd x\, \varphi_\sigma*g_R^2(\gx)|\xi(x)|^2
    = \int_P\rd z\, \varphi_\sigma(\gq)(\xi,\Pi_z\xi).
  \end{equation}
  Furthermore, by Newton's theorem,
  \begin{align}
    \varphi_\sigma*g_R^2
    &=Z\int_{\rz^3}\rd \gy \,\frac{g_R(\gy)^2}{|\gx-\gy|}
      -2D(g_R^2,\sigma)\nonumber\\
    &\geq - \int_{\rz^3}\frac{\sigma(\gy)}{|\gx-\gy|} +
    {\left\{\begin{array}{@{}cl@{}}
      \displaystyle\frac{Z}{|\gx|}&\text{for\ } |\gx|>R\\
      0&\text{for\ }|\gx|\leq R
    \end{array}\right\}},
  \end{align}
  i.e., 
  \begin{equation}
    \label{fi}
    \int_{\Gamma}\rd x \,\varphi_\sigma(\gx)|\xi(x)|^2
    \leq \int_{\Gamma}\rd x\, \varphi_\sigma*g_R^2(\gx)|\xi(x)|^2
    + \int_{R>|\gx|}\rd x\,\frac{Z}{|\gx|}|\xi(x)|^2 .
  \end{equation}
  We write now $\gamma=\sum_n\lambda_n|\xi_n\rl\xi_n|$ in its spectral
  representation with orthonormal eigenvectors
  $\xi_1,\xi_2,\ldots\in H^1(\Gamma)$ and weights
  $0\leq\lambda_1,\lambda_2,\ldots\leq 1$, set $\xi=\xi_n$ in \eqref{T}
  and \eqref{fi}, multiply by $\lambda_n$, and sum over $n$. Continuing
  \eqref{lb} we get for $\epsilon_1,\epsilon_2\in (0,1/2)$
  \begin{align}
    \label{4.55}
    E^\mathrm{S}(Z,N)
    &\geq \int_P\rd z \,{\left[
        (1-\epsilon_1-\epsilon_2)\frac{\gp^2}{2}
      - \varphi_\sigma(\vec{q})\right]}\tr(\gamma\Pi_z) -D[\sigma]
      \nonumber\\
    &\mathrel{\phantom{=}}\mbox{}
    - \frac{1-\epsilon_1-\epsilon_2}{2R^2}\pi^2N +
    \tr{\left(-\frac{\epsilon_1}2\Delta -
      \frac{Z\chi_{B_R(0)}}{|\cdot|}\right)}_-\nonumber\\
    &\mathrel{\phantom{=}}\mbox{}
      +%\inf_{0\leq\gamma\leq1,\ (1-\Delta)\gamma\in\gS^1(\Gamma),\
    % \tr\gamma=N}
    \mathop{\mathop{\inf_{0\leq\gamma\leq1,}}_{(1-\Delta)\gamma\in\gS^1(\Gamma),}}
    _{\tr\gamma=N}\hspace*{-1em}
    \tr\Bigl(\left[-\frac{\epsilon_2}2\Delta -
      \tfrac12\sqrt[3]{9\pi(M\sigma)}\right]\gamma\Bigr)\nonumber\\
    &\geq \int\limits_P\rd z\,\left(
      (1-\epsilon_1-\epsilon_2)\frac{\gp^2}{2}
      - \varphi_\sigma(\vec{q})\right)_- \\\
    &\mathrel{\phantom{=}}\mbox{}
    - D[\sigma] - \const{\left(\frac{N}{R^2} +
      \epsilon_1^{-\frac32}\int\limits_{|\gx|<R}
      {\left(\frac{Z}{|\gx|}\right)}^{\frac52}\right)}\nonumber\\
    &\mathrel{\phantom{=}}\mbox{}    
  + %\inf\limits_{\rho\in L^\frac53(\rz^3),\ \int_{\rz^3}\rho\leq N}
  \mathop{\inf_{\rho\in L^\frac53(\rz^3),}}_{\int_{\rz^3}\rho\leq N}
    \int_{\rz^3}\rd \gx\,
    {\left(\tfrac35\epsilon_2\gamma_\mathrm{LT}\rho(\gx)^\frac53
      -\tfrac12\sqrt[3]{9\pi M(\sigma)(x)}\rho(\gx)\right)}\nonumber
  \end{align}
  where we used $0\leq \gamma\leq1$ and therefore
  $0\leq \tr(\gamma\Pi_z)\leq1$ and we used the Lieb--Thirring
  inequality in both forms \eqref{lt} and \eqref{lt2}.

  \enlargethispage{1.0\baselineskip}

  It is now appropriate to choose $N$ and $\sigma$. We pick ${N=Z}$ and
  $\sigma$ as the minimizer of the first term of the fourth line of
  \eqref{4.55} which amounts to the minimizer of the Thomas--Fermi
  functional but with a different Thomas--Fermi constant, namely
  instead of $(3\pi^2)^\frac23/2$ the constant
  ${\tilde\gamma:=(1-\epsilon_1-\epsilon_2)(3\pi^2)^\frac23/2}$. Rescaling the
  Thomas--Fermi functional, we get that this term becomes
  $(1-\epsilon_1-\epsilon_2)^{-1}E^\mathrm{TF}(Z)$.
  
  Eventually we turn to the last line of \eqref{4.55}. Obviously, the
  maximal function is homogeneous for positive constants, is monotone,
  i.e., $0\leq f\leq g$ implies $M(f)\leq M(g)$. Moreover
  $\sigma(\gx)\leq (Z/(\tilde\gamma|\gx|)^\frac32$ by \eqref{tf}, and
  $1/|\gx|^\frac32$ is an eigenfunction of the maximal operator with
  eigenvalue $C_{3/2,3}$ by \eqref{mf}. Thus, the last line of \eqref{4.55} is
  bounded from below as follows 
  \begin{align}
    \label{Zeile4}
 &\mathrel{\phantom{=}}   
   \mathop{\inf_{\rho\in L^\frac53(\rz^3),}}_{ \int_{\rz^3}\rho\leq N}
   \int_{\rz^3}\rd \gx
    \left(\tfrac35\epsilon_2\gamma_\mathrm{LT}\rho(\gx)^\frac53
      -\tfrac12\sqrt[3]{9\pi M(\sigma)(x)}\rho(\gx)\right)\nonumber\\
  &\geq \mathop{\inf_{\rho\in L^\frac53(\rz^3),}}_{\int_{\rz^3}\rho\leq N}
    \int_{\rz^3}\rd \gx
    \left(\tfrac35\epsilon_2\gamma_\mathrm{LT}\rho(\gx)^\frac53
    -\tfrac12\sqrt[3]{\frac{9\pi C_{3/2,3}}{\tilde\gamma^\frac32}}
    \frac{Z^\frac12}{|\gx|^\frac12}\rho(\gx)\right)\nonumber\\
    &\geq -\const {Z^\frac{13}9/\epsilon_2^\frac13}.
  \end{align}
  Thus
  \begin{equation}
    E^\mathrm{S}(Z)\geq E^\mathrm{TF}(Z) -
    \const{\left[(\epsilon_1+\epsilon_2)Z^{\frac73}+\frac{Z}{R^2}+
      \frac{Z^{\frac52}R^{\frac12}}{\epsilon_1^{\frac32}}+
      \frac{Z^\frac{13}9}{\epsilon_2^\frac13}\right]}.
  \end{equation}
  Optimizing first in $R$ and $\epsilon_1$ yields an error term of the
  order $O(Z^\frac{25}{11})$ followed by an optimization in
  $\epsilon_2$ yielding $O(Z^\frac53)$ we get
  \begin{equation}
    E^\mathrm{S}(Z,Z)\geq E^\mathrm{TF}(Z) -\const Z^\frac{25}{11}
  \end{equation}
  which is the missing lower bound. 
\end{proof}

There are various extensions and related results to the above
asymptotic exactness: 
\begin{description}
\item[\textbf{Scott and Schwinger corrections:}] 

  The next order term in the asymptotic \eqref{energieasymptotik} is
  known. The energy with leading correction is
  \begin{equation}
    \label{scott}
    E^\mathrm{S}(Z)= E^\mathrm{TF}(Z)+ \frac{Z^2}{2} +O\bigl(Z^{\frac{47}{24}}\bigr),
  \end{equation}
  i.e., there is a correction term of order $Z^2$ predicted by and
  named after Scott \cite{Scott1952}. There are various proofs of this
  result. One
  \cite{SiedentopWeikard1987O,SiedentopWeikard1986,SiedentopWeikard1987U,SiedentopWeikard1991}
  that is close to the proof of the leading order uses Macke orbitals
  \eqref{Mackel} instead of coherent states leading to the
  Hellmann--Weizs\"acker functional \eqref{Hellmann} but replaced by
  hydrogenic orbitals for small angular momenta both for the upper and
  the lower bound. The method was actually carried through in the
  noninteracting setting earlier \cite{SiedentopWeikard1986}. There
  were also earlier results on the lower bound using a WKB-type
  analysis (\cite{SiedentopWeikard1989}, Hughes
  \cite{Hughes1986,Hughes1990}). Moreover the result was extended to
  ions (Bach \cite{Bach1989}) and to molecules (Ivrii and Sigal
  \cite{IvriiSigal1993}, Solovej and Spitzer
  \cite{SolovejSpitzer2003}, Balodis \cite{Balodis2004}).

  \rule{2em}{0pt}%
  As Scott conjectured the correction is generated by the Coulomb
  singularity. A nonsingular potential has no correction up to order
  $O(Z^\frac53)$.
  
  \rule{2em}{0pt}%
  In fact even the subleading correction $-\gamma_SZ^\frac53$ was
  predicted by Schwinger \cite{Schwinger1981}
  yielding
  \begin{equation}
    \label{feffermanseco}
    E^\mathrm{S}(Z)= E^\mathrm{TF}(Z)+ \frac{Z^2}{2}-\gamma_SZ^\frac53+o(Z^\frac53)
  \end{equation}
  and established in a series of papers by Fefferman and Seco
  \cite{FeffermanSeco1990O,FeffermanSeco1992,FeffermanSeco1993,FeffermanSeco1994,FeffermanSeco1994T,FeffermanSeco1994Th,FeffermanSeco1995}.

  Unfortunately we cannot present any of those results here as their
  size --- in particular the latter one  --- transcends the limitations
  of a proceedings contribution.
\item[\textbf{The convergence of the density on the Thomas--Fermi
    scale:}] Hand in hand with the asymptotic energetic exactness
  there is also convergence of the quantum density. Suppose that
  $\rho_Z^S$ is a sequence of reduced one-particle densities of ground
  states of $H_{-Z/|\cdot|,Z}$ as $Z\to\infty$. Then
  \begin{equation}
    \int_M\rd x\, Z^{-2}\rho_Z^\mathrm{S}(Z^{-\frac13}x)
    \to\int_M\rd x\,\rho_1(x)
  \end{equation}
  for every bounded measurable set $M$ where $\rho_1$ is the minimizer of the
  Thomas--Fermi functional for ${Z=1}$ (Baumgartner 
  \cite{Baumgartner1976} and Lieb and Simon
  \cite{LiebSimon1977}). This weak convergence with a large set of
  test function can be actually supplemented by a convergence in
  Coulomb norm (see Merz et al.\  \cite{MerzSiedentop2019} in the context
  of relativistic quantum mechanics which, however, holds also in the
  nonrelativistic context).
\end{description}

We conclude with another apparently strange fact besides the
nonexistence of negative ions. It concerns the molecular Thomas--Fermi
functional which is like \eqref{TFF} but with $Z/|\gx|$ replaced by
the potential of all $K$ nuclei $\sum_{k=1}^KZ_k/|\gx-\gR_k|$ and the
nuclear-nuclear repulsion ${\sum_{1\leq k<l\leq K}Z_kZ_l/|\gR_k-\gR_l|}$
added. Teller \cite{Teller1962} showed that the molecular Thomas--Fermi
functional is for all $\rho\in \cI$ and all pairwise different nuclear
positions $\gR_1,\ldots,\gR_K\in\rz^3$ always bigger than the sum of the
atomic energies $E^\mathrm{TF}(Z_1)+\cdots +E^\mathrm{TF}(Z_K)$.  (For a
detailed proof see Lieb and Simon \cite{LiebSimon1977}.)

Although this looks at first sight completely unphysical, it expresses
an important fact: the energy is bounded from below by a linear
quantity in the number of particles involved. Together with the
Lieb--Thirring inequality this is an essential input for showing that the
thermodynamic limit of Coulomb systems exists (Lieb and Lebowitz
\cite{LiebLebowitz1972}), since the energy is an extensive quantity,
i.e., it is proportional to the amount of matter.

\subsection{The Thomas--Fermi--Weizs\"acker functional\label{Weizsacker}}
Already the heuristic derivation of the Thomas--Fermi functional suggests
that the Thomas--Fermi approximation might fail when the potential does
change rapidly. This is also reflected in the proof of Theorem
\ref{satzenergieasymptotik} where the leading error term occurs
because of the nuclear singularity. Weizs\"acker realized this problem
in the context of nuclear physics and suggested to add a term which
scales like the kinetic energy, is rotational and translation
invariant, and penalizes changes of the density, i.e., is nonnegative
and vanishes where the density is constant. He suggested the following
functional
\begin{equation}
  \label{tfw}
  \cE^\mathrm{TFW}_Z(\rho)
  := \frac{\lambda}{2} \int_{\rz^3}\rd\gx\, |\grad\sqrt\rho|^2
  + \cE^\mathrm{TF}_Z(\rho),
\end{equation}
known as the Thomas--Fermi--Weizs\"acker functional, defined on
$\cA$. Weiz\-s\"acker's original choice was $\lambda=1$. A derivation by
the gradient expansion of the Hohenberg--Kohn functional yields
$\lambda=\frac19$ (Kirzhnits \cite{Kirzhnits1957}) whereas adaptation to
numerical computations yields $\frac15$ (Yonei and Tomishima
\cite{YoneiTomishima1965}). We will comment on yet another adaptation
in the context of the Scott correction.

Benguria \cite{Benguria1979} pioneered the mathematical investigation
of the TFW-functional. We will follow --- with minor modifications ---
Benguria et al.\  \cite{Benguriaetal1981} for the basic results: We
remark that $\grad{\sqrt\rho}\in L^2(\rz^3)$ (and $\rho \in L^1(\rz^3)$ and
therefore decaying at infinity) implies $\rho\in L^3(\rz^3)$ by the
Sobolev inequality. Since $\rho\in L^1(\rz^3)$, it follows by
interpolation that $\rho \in L^\frac53(\rz^3)$. By the same argument
$\rho\in L^\frac65(\rz^3)$ and therefore by the
Hardy--Littlewood--Sobolev inequality
$D[\rho] \leq \const \|\rho\|_\frac65<\infty$. Thus $\rho\in \cI$.

This implies that all the terms are well defined and, moreover,
\begin{equation}
  \label{wunten}
  \inf\bigl\{ \cE^\mathrm{TFW}_Z(\cA)\bigr\}
  \geq \inf\bigl\{\cE^\mathrm{TF}(\cI)\bigr\}>-\infty. 
\end{equation}

\begin{lemma}
  The Thomas--Fermi--Weizs\"acker functional and its restrictions to
  $\cA_N$ and $\cA_{\partial N}$ are strictly convex.
\end{lemma}
\begin{proof}
  Obviously the sets $\cA$, $\cA_N$, and $\cA_{\partial N}$ are
  convex. Moreover, the Thomas--Fermi functional is strictly convex by
  Lemma \ref{tfstrengkonvex}. Thus it suffices to show convexity of
  the Weizs\"acker term. To this end pick $\alpha\in(0,1)$ and
  $\rho,\sigma\in\cA$. We set $\psi_1:=\sqrt\rho$,
  $\psi_2:=\sqrt\sigma$, and
  $\psi_3:=\sqrt{\alpha\rho+(1-\alpha)\sigma}$. Then
  \begin{align}
    \psi_3\grad\psi_3
    &= \sqrt\alpha\psi_1\sqrt{\alpha}\grad\psi_1
      +\sqrt{1-\alpha}\psi_2\sqrt{1-\alpha}\grad\psi_2\nonumber\\
    &\leq \sqrt{\alpha \rho+(1-\alpha)\sigma}
    \sqrt{\alpha|\grad\psi_1|^2+(1-\alpha)|\grad\psi_2|^2}
  \end{align}
  by the Schwarz inequality. Thus, we have after integration
  \begin{equation}
  \int_{\rz^3}\Bigl|\grad\sqrt{\alpha\rho+(1-\alpha)\sigma}\Bigr|^2\leq
  \alpha
  \int_{\rz^3}\bigl|\grad\sqrt\rho\bigr|^2
  +(1-\alpha)\int_{\rz^3}\bigl|\grad\sqrt\sigma\bigr|^2    
  \end{equation}
  which is the convexity of the Weizs\"acker term.
\end{proof}
As usual, the strict convexity implies that there is at most one
minimizer of the functional. The existence of the minimizer is shown
again by weak compactness. The argument is more elaborate but in
spirit similar to the Thomas--Fermi case. We skip the proof and merely
report the result:
\begin{theorem}
  The Thomas--Fermi--Weizs\"acker functional has a unique minimizer
  $\rho_Z$ on $\cA$ with $N_c:= \int_{\rz^3}\rho_Z>Z$. Moreover, for
  $N\leq N_c$ it has a unique minimizer $\rho_{Z,N}$ on
  $\cA_{\partial N}$ whereas for $N>N_c$ there is no minimizer on
  $\cA_{\partial N}$. The minimizers fulfill the Euler equations
  \begin{align}
    \label{TFWU}
    \biggl(-\frac\lambda2\Delta
    + \gtf \rho_Z^\frac43-\varphi_{\rho_Z}\biggr)\sqrt{\rho_Z}
    &=0,\nonumber\\  %\label{TFWN}
    \biggl(-\frac\lambda2\Delta
    + \gtf \rho_{Z,N}^\frac43-\varphi_{\rho_{Z,N}}\biggr)\sqrt{\rho_{Z,N}}
    &= -\mu_{Z,N}\sqrt{\rho_{Z,N}}
  \end{align}
  with some $\mu_{Z,N}>0$.
\end{theorem}

Thus, in contrast to Thomas--Fermi theory, there are negative ions in
Thomas--Fermi--Weizs\"acker theory. Benguria and Lieb
\cite{BenguriaLieb1985} gave an estimate on the maximal negative
ionization. They show
\begin{equation}
  \label{tfwuberschuss}
  N_c\leq Z+270.74 \left(\frac{\lambda}{2\gtf}\right)^\frac32
\end{equation}
(actually times the number of of nuclei in the molecular case). Using
$\lambda=\frac15$ as suggested by Yonei and Tomishima
\cite{YoneiTomishima1965} \eqref{tfwuberschuss} becomes
$N_c\leq Z + 0.82$ which is not far from the physical fact that there
are no doubly or higher charged negative ions (Massey
\cite{Massey1976}).

We do not prove the bound \eqref{tfwuberschuss} here. Instead, we
offer the famous --- again unpublished --- argument of Benguria
developed in Thomas--Fermi theory. It gives a worse bound, but is
relatively short and widely used, e.g., it generalizes to the
Hartree--Fock and quantum case where it has been used by Lieb
\cite{Lieb1984}.
\begin{theorem}
  \label{rafw}
  \begin{equation}
    \label{RafaelW}
    N_c<2Z.
  \end{equation}
\end{theorem}
\begin{proof}
  We first we note an observation by Lieb \cite{Lieb1984}, namely 
  \begin{equation}
    \label{LiebHardy}
    -\Delta|\cdot|+|\cdot|(-\Delta)>0.
  \end{equation}
  We prove this inequality by recognizing that it is merely a recast
  of Hardy's inequality. For, say, $f\in C_0^\infty(\rz^3)$, we have
  \begin{align}
       -\bigl(f,(\Delta|\gx|+|\gx|\Delta)f\bigr)
    &= {\left(f,\left\{\sqrt{|\cdot|}(-\Delta)\sqrt{|\cdot|}
      +\left[\sqrt{|\cdot|},\left[\sqrt{|\cdot|},
      -\Delta\right]\right]\right\}f\right)}\nonumber\\
    &={\left(\sqrt{|\cdot|}f,
      \left(-\Delta-\frac{1}{4|\cdot|^2}\right)\sqrt{|\cdot|}f\right)}>0
  \end{align}
  unless $f=0$. (Here $[A,B]:=A\circ B-B\circ A$ denotes the commutator.)

  In passing we note that \eqref{LiebHardy} can be generalized to the
  relativistic case, i.e., $\sqrt{-\Delta}$ instead of $-\Delta$ (see
  Lieb \cite{Lieb1984} for a restricted validity, Dall'Acqua and
  Solovej \cite{Dall'AcquaSolovej2010} for the full inequality and
  Handrek et al.\  \cite{HandrekSiedentop2013} for an extension to the
  two-dimensional case.) In fact a much larger class of such
  anticommutator inequalities is true (Chen et al.\ 
  \cite{ChenSiedentop2013}).

  We prove \eqref{RafaelW} by multiplying \eqref{TFWU} by
  $\sqrt{\rho(\gx)}|\gx|$ followed by integration:
  \begin{align}
    \label{mulmitx}
  0&=\int_{\rz^3}\!\!\rd \gx\,\biggl( -|\gx|\sqrt{\rho_Z(\gx)}
    \frac\lambda2\Delta\sqrt{\rho_Z(\gx)}
    + \gtf|\gx|\rho_Z(\gx)^\frac53
    -|\gx|\rho_Z(\gx)\varphi_{\rho_Z}(\gx)\biggr)\nonumber\\
    & > -Z\int_{\rz^3}\rho_Z + \int_{\rz^3}\rd \gx\int_{\rz^3}\rd \gy\,
      \frac{|\gx|}{|\gx-\gy|}\rho_Z(\gx)\rho_Z(\gy) \nonumber\\
    & = -Z\int_{\rz^3}\rho_Z + \frac12\int_{\rz^3}\rd \gx\int_{\rz^3}\rd \gy\,
      \frac{|\gx|+|\gy|}{|\gx-\gy|}\rho_Z(\gx)\rho_Z(\gy)\geq -Z N_c +\frac12N_c^2
  \end{align}
  where we have used that the first summand is real allowing to
  replace $|\cdot|(-\Delta)$ by half of its anticommutator, and then
  use the inequality \eqref{LiebHardy}. In the last
  step, we used the triangle inequality and
  $\int_{\rz^3} \rho_Z=N_c$. Rearranging terms gives the desired
  inequality.
\end{proof}

To conclude this section we expand on the above remark about the
wide applicability of Benguria's idea: a similar type of argument
shows that any minimizer $\gamma$ of the Hartree--Fock functional on
$\cD$ (see Section \ref{Abschnitt:HF}) fulfills the inequality
\begin{equation}
  \label{RafaelH}
  \tr (\gamma)<2Z+1
\end{equation}
and that the Hamiltonian $H_{-Z/|\cdot|,N}$ has no bound states, if
$N\geq2Z+1$ (Lieb \cite{Lieb1984}). The latter implies that there is no
doubly charged negative hydrogen.

\subsection{The Engel--Dreizler functional}

The known limiting theorem, as well as the heuristic derivation, of
Thomas--Fermi and related theories assume large particle
numbers. This, in turn requires for atoms also large atomic number
$Z$, since, as we saw, only large $Z$ atoms allow for large number $N$
of electrons to be bound. However, this renders a nonrelativistic
treatment physically questionable: Using $\tfrac m2 v^2=-E$ and using
$E= -Z^2/2$, i.e., the hydrogenic ground state energy, yields an
estimated speed of $Z$ for the innermost electrons. This compares with
the speed of light $c=137.037$ in Hartree units (Michelson
\cite{Michelson1927}). Thus, for uranium, ${Z=92}$, the innermost
electrons reach roughly $0.67$ of the speed of light. This suggests
that relativistic effects should be taken into account which has been
realized early. However, the heuristic derivation along the lines of
the derivation of the nonrelativistic Thomas--Fermi functional yields
an atomic functional which is unbounded from below. The semiclassical
relativistic kinetic energy term is not strong enough to prevent
collapse of electrons into the nucleus. A review of these facts and
various ad hoc attempts to cure this problem can be found in Gombas's
classical book \cite{Gombas1949} and encyclopedia article
\cite{Gombas1956}.

Engel and Dreizler \cite{EngelDreizler1987} offered a systematic
solution to the problem. The functional which they propose reads in
the atomic case
\begin{equation}
\label{TFWD}
\cE^\mathrm{ED}_Z(\rho)
:= \W(\rho)+\TF(\rho)-\X(\rho) +\V(\rho). 
\end{equation}
The first summand on the right is an inhomogeneity correction of the
kinetic energy generalizing the Weizs\"acker correction.  Using the
abbreviation $p(\gx):= \bigl(3\pi^2 \rho(\gx)\bigr)^{1/3}$,
\begin{equation}
  \label{W}
  \W(\rho):=
  \int_{\rz^3}\rd\gx\,\frac{3\lambda}{8\pi^2}
  \bigl(\grad p(\gx)\bigr)^2c\,f\bigl(p(\gx)/c\bigr)^2
\end{equation}
with ${ f(t)^2:=t(t^2+1)^{-\frac12}+2t^2(t^2+1)^{-1}\arsinh(t)}$ where
$\arsinh$ is the inverse function of the hyperbolic sine and
$\lambda\in\rz_+$ is given by the gradient expansion as $\frac19$ but in
the nonrelativistic analog sometimes taken as an adjustable
parameter (Weizs\"acker \cite{Weizsacker1935}, Yonei and Tomishima
\cite{YoneiTomishima1965}, Lieb and Lieberman \cite{LiebLiberman1982}, Lieb
\cite{Lieb1982A}). 
The second summand is the relativistic generalization of the
Thomas--Fermi kinetic energy. It is
\begin{equation}
  \label{TF}
  \TF(\rho):=\int_{\rz^3}\rd\gx\,\frac{c^5}{8\pi^2}
  \tf\bigl(p(\gx)/c\bigr) 
\end{equation}
with
${\tf(t):=t(t^2+1)^{3/2}+t^3(t^2+1)^{1/2}-\arsinh(t)-\frac{8}{3}t^3}$.
The third summand is a relativistic generalization of the exchange
energy. It is
\begin{equation}
  \label{X}
  \X(\rho):= \int_{\rz^3}\rd \gx\,\frac{c^4}{8\pi^3} X\bigl(p(\gx)/c\bigr)
\end{equation}
with ${X(t):= 2t^4-3[t(t^2+1)^\frac12-\arsinh(t)]^2}$, and, eventually,
the last summand is the potential energy, namely the sum of the
electron-nucleus energy and the electron-electron energy. It is
\begin{equation}
  \V(\rho):= -Z\int_{\rz^3}\rd \gx\, \rho(\gx)|\gx|^{-1}+D[\rho].
\end{equation}

The Engel--Dreizler functional is defined on
\begin{equation}
  \label{domain}
  P:=\bigl\{\rho\in L^\frac43(\rz^3)\bigm|
  \rho\geq0,\ D[\rho]<\infty, F\circ p\in D^1(\rz^3)\bigr\}.
\end{equation}

Initial steps analyzing the functional have been carried out by Chen
et al.\  \cite{ChenSiedentop2020} and refined in \cite{Chenetal2020} (in
these proceedings). It is bounded from below for all nuclear charges
-- a surprising fact which is not even true for various quantum models
--, it yields in first order for $c/Z$ fixed the Thomas--Fermi theory,
and is stable in the sense of stability of matter
\cite{Siedentop2021}. A relativistic correction is expected for the
Scott correction of the functional.

\subsection{Density functionals in phase space}
The functionals of the one-particle density $\rho$ in position space
which we have considered so far can be generalized to functionals of
the one-particle density $f$ in phase space. They are of particular
interest, if the Hamiltonian of the system is no longer a sum of terms
which depend either on momentum or on position but a sum that also
contains mixed terms. Densities depending on both momentum and
position are also important, if the time dependence of the density is
of interest.

The set of measurable functions $f:\rz^6\to\rz_+$ is the set of
(semiclassical) densities in phase space. For a system of identical
fermions of $q$ spin states each the Pauli principle is
implemented by requiring the additional constraint $f\leq q$. The
integral $\int_{\rz^3}\rd\gx\int_{\rz^3}\dbar\gxi\, f(\gx,\gxi)$ is the
particle number of $f$. The fermionic phase space energy functional of
an atom of a fermionic quantum system with kinetic energy $T$,
external potential $V$, additional momentum dependent external
potential $\tilde V$, and interaction $W$ is
\begin{align}
  \label{spf}
  \cE(f)&:= \int_{\rz^6}\dbar \gxi\, \rd \gx\,
  \bigl[T(\gxi)+V(\gx)+\tilde V(\gxi)\bigr]f(\gxi,\gx)\nonumber\\
  &\mathrel{\phantom{=}}\mbox{}
  + \frac12\int_{\rz^6}\dbar \gxi\, \rd \gx\int_{\rz^6}\dbar
  \geta\,\rd \gy \,W(\gx,\gy)f(\xi,\gx)f(\geta,\gy)
\end{align}
with $f:\rz^6\to[0,q]$. We call
$\rho_f(\gx):= \int_{\rz^3}\dbar \gxi\, f(\gxi,\gx)$ the position density at
$\gx$ and $\tau_f(\gxi):= \int_{\rz^3}\dbar \gx\, f(\gxi,\gx)$ the
momentum density at $\gxi$.

This is, of course, formal only as long as we do not specify the
various terms and the allowed densities and show that all terms are
well defined. For definiteness and simplicity we will again consider a
generic choice only, namely an atom, i.e., we choose
$T(\gxi):=\gxi^2/2$, $V(\gx):=-Z/|\gx|$, $\tilde V=0$,
$W(\gx,\gy)=1/|\gx-\gy|$, $q=2$, and the domain of definition
\begin{equation}
  \mathcal{P}
  := \biggl\{f:\rz^6\to[0,2]\biggm| \int_{\rz^6}\dbar\gxi\,\rd\gx\,
  \bigl(1+|\gxi|^2\bigr)f(\gxi,\gx)<\infty,\ \rho_f\in\cC\biggr\}.
\end{equation}

The kinetic energy term is finite by definition as well as the
electron-electron repulsion, since it is equal to $D[\rho_f]$. We
decompose the electron-nucleus attraction as in \eqref{zerlegung}. The
long-range part $V_l$ is dominated as in Thomas--Fermi theory by
$D[\rho_f]$. It remains to control the short-range part:
\begin{align}
  &\mathrel{\phantom{=}}
    \int_{\rz^3}\rd\gxi\int_{|\gx|< R}\rd\gx \,
    \frac{1}{|\gx|\bigl(1+|\gxi|^2\bigr)^{\alpha}}
    \bigl(1+|\gxi|^2\bigr)^\alpha f(\gxi,\gx)\nonumber\\
  &\leq {\left(\int_{|\gx|<R} \frac{\rd \gx}{|\gx|^p}
    \int_{\rz^3}\frac{\dbar\gxi}{\bigl(1+|\gxi|^2\bigr)^{\alpha p}}\right)}^{\frac1p}
  {\left(\int_{\rz^6}\dbar\gxi\,\rd\gx\, \bigl(1+|\gxi|^2\bigr)^{\alpha q}
    f(\gxi,\gx)^q\right)}^{\frac1q}
  \nonumber\\
  &\leq 2^{\frac{10}{29}}{\left(\int_{|\gx|<R}
    \frac{\rd \gx}{|\gx|^\frac{29}{10}} \int_{\rz^3}
    \frac{\dbar\gxi}{\bigl(1+|\gxi|^2\bigr)^{\frac{19}{10}}}\right)}^{\frac{10}{29}}
    {\left(\int_{\rz^6}\dbar\gxi\,\rd\gx\,\bigl(1+|\gxi|^2\bigr)
    f(\gxi,\gx)\right)}^{\frac{19}{29}}\nonumber\\
  &<\infty
\end{align}
with ${p=\frac{29}{10}}$, $q={\frac{2}{19}}$, and ${\alpha =\frac{1}{q}}$.
This shows that all the terms are well defined.

\subsubsection{Marginal functionals: The position space}
One possibility to minimize $\cE$ is to split the minimization into two
steps. The classical way is to fix $\gx$ and to minimize --- in a first
step --- over the functions $f(\cdot,\gx)$. Any minimizer can obviously
be picked such that it takes the values $0$ and $2$ only and because
of the symmetry and the monotony of the Coulomb potential they can be
picked as characteristic functions (times $2$) of balls
$B_{R(\gx)}(0)$ centered at the origin with radius $R(\gx)$ aka Fermi
balls and Fermi radius. The relation between the radii and the density
is $R(\gx)=\sqrt[3]{3\pi^2\rho_f(\gx)}$. This yields
\begin{equation}
  \label{scgtf}
  \cE(f)\geq \cE^\mathrm{TF}_Z(\rho_f).
\end{equation}
The second step consists in minimizing the Thomas--Fermi functional as
described in Sec.~\ref{tftheorie}. Thus, the minimizer
$f_\mathrm{min}$ of $\cE$ on $\mathcal{P}$ is --- apart from the factor 2 ---
the characteristic function of the the classical allowed phase space
\begin{equation}
\bigl\{(\gxi,\gx)\in\rz^6\bigm|\gxi^2/2-\varphi_{\rho_Z}(\gx)<0\bigr\}  
\end{equation}
for the Thomas--Fermi potential $\varphi_{\rho_Z}$, i.e.,
\begin{equation}
  f_\mathrm{min}(\gxi,\gx)
  := 2\theta\bigl(-\bigl[\gxi^2/2-\varphi_{\rho_Z}(\gx)\bigr]\bigr).
\end{equation}

This shows that the infimum of the phase space functional $\cE$ is
equal to the infimum of the Thomas--Fermi functional. Thus $\cE$ is not
only well defined but also bounded from below yielding the same
energy.

\subsubsection{Marginal functions: The momentum space}
A less obvious way to minimize $\cE$ is to reverse the order, i.e., to
fix $\gxi$ in the first step and to minimize $\cE$ over the functions
$f(\gxi,\cdot)$. Again those minimizers take only the values $0$ and
$2$ and are characteristic functions of balls $B_{\tilde R(\gxi)}(0)$
centered at the origin with a now $\gxi$-dependent radius
$\tilde R(\gxi)$ with the relation
$\tilde R(\gxi)=\sqrt[3]{3\pi^2\tau(\gxi)}$ to the momentum density
$\tau$. The result is Englert's momentum space semiclassical
functional \cite{Englert1992}
\begin{align}
  \cE^\mathrm{E}_Z(\tau)
  &:= \int_{\rz^3} \rd\gxi\,
    \gxi^2\tau(\gxi)
    -  \frac32 \frac{Z}{\sqrt[3]{3\pi^2}}\int_{\rz^3}
    \rd\gxi\,\tau(\gxi)^{2/3}
    \nonumber\\
  &\mathrel{\phantom{=}}\mbox{}
  + \frac34\frac{1}{\sqrt[3]{3\pi^2}} \int_{\rz^3} \rd \gxi \int_{\rz^3} \rd
    \geta\, \bigl(\tau_<(\gxi,\geta)\tau_>(\gxi,\geta)^{2/3}-\tfrac15
    \tau_<(\gxi,\geta)^{5/3}\bigr)
\label{ImpFunk}
\end{align}
where $\tau_<(\gxi,\geta):= \min\{\tau(\gxi),\tau(\geta)\}$ and
$\tau_>(\gxi,\geta):= \max\{\tau(\gxi),\tau(\geta)\}$.

The second step is the minimization of $\cE_Z^\mathrm{E}$. As to be
expected it yields the same energy as the Thomas--Fermi functional.

A mathematical analysis of $\cE^\mathrm{E}_Z$ including the convergence
of the one-particle reduced momentum density of an atom described by
the Schr\"odinger equation to the minimizer of the hydrogenic Englert
functional has been carried through by Conta et
al.\ \cite{ContaSiedentop2015}.

At first sight $\cE^\mathrm{E}$ might look repelling: whereas the
kinetic energy term is simple, the two potential terms, in particular
the electron-electron repulsion term look unfamiliar. However, when
momentum-dependent potentials are present like for the Compton profile,
it simplifies the treatment. Another application is an alternate
investigation of the Scott correction (Cinal and Englert
\cite{CinalEnglert1992,CinalEnglert1993E}).

\subsubsection{Time-dependent equations \label{Wlassowundco}}
Runge and Gross \cite{RungeGross1984} argued for the existence of a
time-dependent analog of the stationary density functional
formalism. In particular they derived an Euler-type equation for the
density and the current. Fournais et al.\  \cite{Fournaisetal2016}
analyzed the argument pointing out various problems with singular
potentials, in particular Coulomb potentials. Nevertheless an equation
of Runge--Gross type generalizing Thomas--Fermi theory is known in the
literature since long --- although to some extent ignored by the
density functional community, e.g., Runge and Gross do not reference
it: Bloch \cite{Bloch1933} (see also Gombas \cite[\S
20]{Gombas1949}) introduced the following Euler-type equation
\begin{eqnarray}
   \label{eq:gtf}
   \partial_t\varphi_t=\frac12(\grad \varphi_t)^2 +\int\frac{\rd p}{\rho_t}
   - \frac{Z}{|\cdot|} + \rho_t*|\cdot|^{-1}  
 \end{eqnarray}
 supplemented by the continuity equation
 \begin{equation}
   \label{eq:kon}
   \partial_t\rho_t=\grad\cdot(\rho_t\grad\varphi_t).
 \end{equation}
 Here $\varphi$ is the potential of the velocity field $\gu$, i.e.,
 $\gu=-\grad \varphi$, $\rho$ is the density of electrons, and $p$ is
 the pressure as a function of $\rho$. The Thomas--Fermi choice for $p$
 is $p(\rho):=\frac15\gtf \rho^{5/3}$,
 i.e., we have
 \begin{equation}
   \label{eq:tf}
   \partial_t\varphi_t=\frac12(\grad \varphi_t)^2 + \gtf\rho_t^{2/3}
   - \frac{Z}{|\cdot|} + \rho_t*|\cdot|^{-1}
 \end{equation}
 which reduces to the stationary Thomas--Fermi equation \eqref{tf} when
 $\varphi$ is constant.

 There are existence results for such Euler equations with regular
 potentials. The Coulomb case, however, seems to be still open. There
 are, however, nonexistence results for highly charged negative
 ions. Chen et al \cite{ChenSiedentop2018}, following a strategy
 developed by Lenzmann and Lewin \cite{LenzmannLewin2013}, showed
\begin{theorem}
   \label{th:tf}
   Assume that $\varphi_t$ and $\rho_t$ is a weak solution of
   \eqref{eq:tf} and \eqref{eq:kon} with
   ${\int_{\rz^3}\rd\gx\,\tfrac12\rho_0(\gx)|\grad\varphi_0(\gx)|^2
   +\cE_\mathrm{TF}(\rho_0)<\infty}$, assume ${B\subset\rz^3}$ bounded
   and measurable, and set
   \begin{equation}
     N_\mathrm{TF}(t,B):= \int_B\rd \gx\,\rho_t(\gx)
   \end{equation}
  which is the number of electrons in $B$.  Then, in temporal average
  for large time, $N_\mathrm{TF}(t,B)$ does not exceed $4Z$, i.e.,
  \begin{equation}
    \label{behauptung-tf}
    \limsup_{T\to\infty} \frac{1}{T}\int_0^T\rd t \,N_\mathrm{TF}(t,B) \leq 4Z.
  \end{equation}
\end{theorem}
This shows that the number of electrons that an atom can acquire is
bounded by $4Z$ generalizing the bound of the time-independent
setting. There is, however, a price paid: the generalization yields a
worse constant.

Whereas the generalization of the Thomas--Fermi equation to the
time-dependent case requires the introduction of the velocity field of the matter,
the transfer of the phase space theory as condensed in the variational
principle \eqref{spf} depending on the phase space density does not
require any extra functions. The natural generalizations is the Vlasov
equation (Vlasov \cite{Vlasov1938,Vlasov1968}). We formulate it --- for
notational simplicity --- in the atomic case only. We write
\begin{equation}
\gK(\gx) := -\grad V_\mathrm{tot}(\gx)= -Z\frac{\gx}{|\gx|^3} 
  + \int_{\rz^3}\rd \gy\,\rho_t(\gy) \frac{\gx-\gy}{|\gx-\gy|^3}
\end{equation}
for the force of the nucleus and the electron cloud exerted on an
electron. With this notation the Vlasov equation reads
  \begin{equation}
  \label{VPN}
  \partial_t f_t + \gxi\cdot \grad_{\gx} f_t + \gK\cdot \grad_{\gxi} f_t  =0.
\end{equation}

Suppose that $f_t$ is a weak solution of the Vlasov equation. Then
\begin{equation}\label{Energy}
  \cE_V(f_t) := \int_{\rz^6}\dbar\gxi \,\rd\gx\,
  \frac12 \gxi^2f_t(\gxi,\gx) 
  -\int_{\rz^3}\rd\gx\,V(\gx)\rho_t(\gx) + D[\rho]
\end{equation}
is called the energy. It is a conserved quantity.
 
As in the case of the time-dependent Thomas--Fermi theory Chen et
al.\ \cite[Theorem 1]{ChenSiedentop2018} proved --- again following
\cite{LenzmannLewin2013} --- an upper bound on the charge that can
remain in any fixed ball:
\begin{theorem}
  \label{maxion}
  Assume $f_t$ to be a weak solution of the Vlasov equation
  \eqref{VPN} of finite energy \eqref{Energy}, assume $B\subset\rz^3$
  bounded and measurable, and set
  \begin{equation}
  N_V(t,B):= \int_{\rz^3} \dbar\gxi \int_B\rd\gx\,f_t(\gxi,\gx)  
  \end{equation}
  which is the number of electrons in $B$.  Then in temporal average
  for large time $N_V(t,B)$ does not exceeds $4Z$, i.e.,
  \begin{equation}
    \label{behauptung}
    \limsup_{T\to\infty} \frac{1}{T}\int_0^T\rd t N_V(t,B) \leq 4Z.
  \end{equation}
\end{theorem}

Actually, the solution of the Vlasov equation, if existing, can be
compared with with the solution of the Schr\"odinger equation when the external
potential is not too singular. (See, e.g., Petrat and Pickl
\cite{PetratPickl2016} and the references given there.)

\section{Functionals of the one-particle density matrix}
\subsection{The Hartree--Fock functional\label{Abschnitt:HF}}
We call
\begin{align}
  \cD&:= \bigl\{\gamma\in\gS^1\bigl(L^2(\Gamma)\bigr)\bigm|0\leq\gamma\leq1,\
          (-\Delta+1)\gamma\in\gS^1\bigl(L^2(\Gamma)\bigr)\bigr\},\nonumber\\
  \cD_N&:= \bigl\{\gamma\in\cD\bigm|\tr(\gamma)\leq N\bigr\},\nonumber\\
  \cD_{\partial N}&:= \bigl\{\gamma\in\cD\bigm|\tr(\gamma)= N\bigr\} \label{DDN}
\end{align}
the set of one-particle reduced density matrices with finite kinetic
energy, those with particle number not exceeding $N$, and those with
particle number equal to $N$.

We call
\begin{align}
  \label{hf}
    &\cE_Z^\mathrm{HF}:\cD\to\rz,\nonumber\\
    &\gamma\mapsto \tr\biggl(-\frac12\Delta\gamma\biggr)
    -\int_{\rz^3}\rd \gx\,\frac{Z\rho_\gamma(\gx)}{|\gx|}+ D[\rho_\gamma]
      -\underbrace{\frac12\int_\Gamma\rd x\int_\Gamma\rd y\,
      \frac{|\gamma(x,y)|^2}{|\gx-\gy|}}_{=:X[\gamma]}
\end{align}
the Hartree--Fock functional where $\rho_\gamma$ is the one-particle
density of $\gamma$ defined as in \eqref{dichte} and $\gamma(x,y)$ is
the density matrix's $\gamma$ integral kernel. The first summand is
the kinetic energy, the second is the attraction potential of the
nucleus, the third is the classical Coulomb potential of the charge
density $\rho_\gamma$, aka direct or Hartree term, the last term
is called the exchange energy of $\gamma$.

The functional is well defined which is obvious for the kinetic energy
by definition. For the second term we use the variational principle
for hydrogenic atoms: for every $\xi\in H^1(\Gamma)$,
$\ga\in \rz^3$ and $\zeta\in(0,\infty)$, we have
\begin{equation}
  \label{hyd}
  \int_\Gamma\rd x\,\frac{|\xi(x)|^2}{|\gx-\ga|}
  \leq \frac12{\left(\zeta^{-1}\int_\Gamma\rd x\,|\grad\xi(x)|^2
    + \zeta \int_\Gamma\rd x\,|\xi(x)|^2\right)}
\end{equation}
and thus with $\zeta=1$
\begin{equation}
  0\leq \int_{\rz^3}\rd \gx\,\frac{\rho_\gamma(\gx)}{|\gx|}
  \leq\sum_{\sigma,\,n}\lambda_n
  \int_\Gamma\rd x\,\frac{|\xi_n(x)|^2}{|\gx|}
  \leq \frac12\bigl(\tr(-\Delta\gamma)+\tr(\gamma)\bigr)<\infty
\end{equation}
rendering the attraction potential finite. For the third term --- using
again \eqref{hyd} with $\zeta=1$ and $\ga=\gy$ --- we have
\begin{equation}
  D[\rho_\gamma]\leq \frac12 \tr(\gamma)\frac12\bigl(\tr(-\Delta\gamma)
  +\tr(\gamma)\bigr)<\infty.
\end{equation}
Eventually we remark that the exchange term is bounded by the direct
term using Schwarz's inequality in the summation over $n$ and $m$:
\begin{align}
  \label{xld}
  0\leq X[\gamma]
  &=\frac12\int_\Gamma\rd x\int_\Gamma\rd y
    \,\frac{\sum_{n,m}\lambda_n\lambda_m\xi_n(x)
    \overline{\xi_n(y)}\,\overline{\xi_m(x)}\xi_m(y)}{|\gx-\gy|}\nonumber\\
  &\leq\frac12\int_{\rz^3}\rd \gx\int_{\rz^3}\rd \gy
  \,\frac{\rho_\gamma(\gx)\rho_\gamma(\gy)}{|\gx-\gy|}=D[\rho_\gamma]
\end{align}
which shows that also the exchange term is finite.

We will now exhibit some properties of the Hartree--Fock theory and
begin with a theorem by Lieb \cite{Lieb1981V} but will
present Bach's \cite[Section 3]{Bach1992} version of the proof.
\begin{theorem}
  \label{Volker}
  Suppose $N\in\mathbb{N}$, $\gamma\in \cD_{\partial N}$, and
  $\gamma\neq\gamma^2$. Then there exists a
  $\gamma'\in \cD_{\partial N}$ such that
  $\cE^\mathrm{HF}_Z(\gamma')<\cE^\mathrm{HF}_Z(\gamma)$. In
  particular, any minimizer of $\cE_Z^\mathrm{HF}$ on
  $\cD_{\partial N}$ is a projection.
\end{theorem}
\begin{proof}
  Since $\gamma$ is not a projection there is at least one eigenvalue,
  say ${\lambda_1\in(0,1)}$. Then, since $N$ is an integer, there is also 
  a second eigenvalue ${\lambda_2\in(0,1)}$. We write $\xi_1$ and
  $\xi_2$ for two corresponding orthonormal eigenvectors and set
  \begin{equation}
  \gamma_\epsilon := \gamma+
  \epsilon\bigl(\underbrace{|\xi_1\rl\xi_1|-|\xi_2\rl\xi_2|}_{=:\delta}\bigr)  
  \end{equation}
  which is in $\cD_{\partial N}$ as long as
  $\epsilon+\lambda_1,-\epsilon+\lambda_2\in[0,1]$,
  $\epsilon\in\rz$. We compute
  \begin{align}
    \label{D-X}
    D[\rho_\delta]-X[\delta]
    &= \int_{\Gamma^2}\frac{\rd x\,\rd y}{|\gx-\gy|}
   \Bigl[\bigl(|\xi_1(x)|^2-|\xi_2(x)|^2\bigr)\bigl(|\xi_1(y)|^2-|\xi_2(y)|^2\bigr)
      \nonumber\\
    &\rule{70pt}{0pt}\mbox{}
      - \bigl|\xi_1(x)\overline{\xi_1(y)}-\xi_2(x)\overline{\xi_2(y)}\bigr|^2\Bigr]
    \nonumber\\
    &= -2\int_{\Gamma^2}\!\!\rd x\,\rd y
    \,\frac{|\xi_1(x)|^2|\xi_2(y)|^2
      - \xi_1(x)\overline{\xi_2(x)}\,\overline{\xi_1(y)}\xi_2(y)}{|\gx-\gy|}<0
  \end{align}
  where, in the last step, we have used the Schwarz inequality and the
  fact that equality in the Schwarz inequality can only hold when
  $\xi_1\otimes\xi_1$ and $\xi_2\otimes\xi_2$ are linearly dependent
  which, however, is definitely not the case, since $\xi_1$ is
  orthogonal to $\xi_2$.  Thus
  \begin{align}
    \label{differenz}
    \cE^\mathrm{HF}_Z(\gamma_\epsilon) -\cE^\mathrm{HF}_Z(\gamma)
    &= \epsilon
      {\left[\tr\Biggl(\biggl(-\frac12\Delta-\frac{Z}{|\gx|}\biggr)\delta\Biggr)
      +2D(\rho_\gamma,\rho_\delta)-2X(\gamma,\delta)\right]}
      \nonumber\\
    &\mathrel{\phantom{=}}\mbox{}
     +\epsilon^2 \bigl(D[\rho_\delta]-X[\delta]\bigr)<0
  \end{align}
  where we choose the sign of $\epsilon$ equal to minus the sign of
  the bracket in \eqref{differenz}, use \eqref{D-X}, and decrease or
  increase $\epsilon$, depending on our choice of the sign of
  $\epsilon$ until $\epsilon+\lambda_1$ or $-\epsilon+\lambda_2$
  become $0$ or $1$ for the first time. This does not only prove the
  theorem but it even shows that we can choose $\gamma'$ such that it
  has at least one less eigenvalue in the open interval $(0,1)$. Of
  course this argument can be iterated as long as there are
  eigenvalues that are not equal to $0$ or $1$.
\end{proof}

This result can be used to show that the particle number of minimizers
is an integer, poetically speaking, is quantized (Friesecke
\cite{Friesecke2003T}).  Writing $G_Z$ for the set of all minimizers
of $\cE_Z^\mathrm{HF}$ on $\cD$, we have
\begin{theorem}
  For any $Z>0$ the set of minimizers $G_Z$ contains a
  projection. Moverover, if for all $\gamma\in G_Z$ the operator
  $h^\mathrm{HF}_{Z,\gamma}$ has no zero eigenvalue, then $G_Z$
  consists of projections only.
\end{theorem}
\begin{proof}
  Running through the proof of Theorem \ref{Volker} again shows that a
  minimizer $\gamma$ can have at most one eigenvalue, say $\lambda$,
  which is strictly between zero and one. Write $\xi$ for a
  corresponding normalized eigenfunction which is also an
  eigenfunction of $h^\mathrm{HF}_{Z,\gamma}$. Such a choice is
  possible, since the Hartree--Fock equations \eqref{hf-gleichung}
  ensure that the Hartree--Fock operator $h^\mathrm{HF}_{Z,\gamma}$
  commutes with the $\gamma$.

  We set
  $\gamma_\epsilon:=\gamma+\epsilon|\xi\rl\xi|$ and write
  \begin{equation}
    \label{hfop}
    h^\mathrm{TF}_{Z,\gamma}f(x)
    := \left(-\frac12\Delta-\frac{Z}{|\gx|}
      -\int_{\rz^3}\frac{\rd\gy\,\rho_\gamma(\gy)}{|\gx-\gy|}\right)f(x)
    - \int_\Gamma \rd y\,\frac{\gamma(x,y)}{|\gx-\gy|}f(y).
    \end{equation}
  Now, we compute
  \begin{equation}
    \cE^\mathrm{HF}_Z(\gamma_\epsilon)-\cE^\mathrm{HF}_Z(\gamma)
    = \epsilon\, \bigl(\xi, h^\mathrm{HF}_{Z,\gamma}\xi\bigr)\leq 0
  \end{equation}
  by choosing the appropriate sign of $\epsilon$.  By choosing
  $\epsilon+\lambda=1$, if the expectation is negative and
  $\epsilon+\lambda=0$ gives a projection that has an energy which is
  at least as low as the energy of $\gamma$. Thus, we do not increase
  the energy -- in fact we strictly decrease the energy unless $\xi$
  is a zero energy eigenvalue of Hartree--Fock operator
  $h^\mathrm{HF}_{Z,\gamma}$ --- by replacing the eigenvalue $\lambda$
  by $0$ or $1$ depending on the sign of
  $(\xi, h^\mathrm{HF}_{Z,\gamma}\xi)$. However, the strict decrease
  would imply that $\gamma$ cannot be a minimizer. Thus, we are left
  with the case of an zero energy eigenvalue, the change in the energy
  is indifferent towards the $\epsilon$. By construction
  $\gamma_\epsilon$ is a minimizer and a projection.
\end{proof}

The Euler equation, i.e., the equation that a minimizer $\gamma$ of the
Hartree--Fock functional $\cE^\mathrm{HF}_Z$ fulfills, reads
\begin{equation}
  \label{hf-gleichung}
  \bigl[h^\mathrm{HF}_{Z,\gamma},\gamma\bigr]=0
\end{equation}
or in orbital form
\begin{equation}
  \label{hf-klassisch}
  h^\mathrm{HF}_{Z,\gamma}\xi_n=-\mu_n\xi_n
\end{equation}
with Lagrange parameters $\mu_1,\ldots,\mu_N\in\rz_+$,
$\xi_1,\ldots,\xi_N$ orthonormal spinors, and
$\gamma=|\xi_1\rl\xi_1|+\cdots +|\xi_N\rl\xi_N|$. The existence of
minimizers, and thus the existence of solutions of the Hartree--Fock
equations, has been shown for $N<Z+1$ by Lieb and Simon
\cite{LiebSimon1977T}. Note, however, that, in general, one cannot
expect uniqueness of the solution because of the lack of convexity
of the exchange term.

\enlargethispage{1.6\baselineskip}%%++

Next we turn to a result that is well known for the traditional
Hartree--Fock functional, i.e., the one with all the $\lambda_n$ equal
to either $0$ or $1$.
\begin{theorem}
  Pick $N\in\mathbb{N}$. If the infimum of
  $\cE^\mathrm{HF}_Z(\cD_{\partial N})$ is assumed, then we have for
  any $\gamma\in\cD_N$
  \begin{equation}
    \label{os}
      E^\mathrm{S}(Z,N)\leq \cE^{HF}_Z(\gamma).
    \end{equation}
  \end{theorem}%\enlargethispage{1.2\baselineskip}%%++
  \begin{proof}
    By Theorem \ref{Volker} we know that $\cE^\mathrm{HF}(\gamma)$ can
    be lowered unless $\gamma$ is a projection and
    $\inf\bigl\{\cE^\mathrm{HF}_Z(\cD_{\partial N})\bigr\}$ is assumed
    by a projection. Thus, we can assume $\gamma$ to be a projection
    in $\cD_{\partial N}$. Thus
    \begin{equation}
      \gamma=|\xi_1\rl\xi_1|+\cdots+ |\xi_N\rl\xi_N|
    \end{equation}
    for orthonormal $\xi_1,\ldots,\xi_N$ each in $H^1(\Gamma)$. We write
    \begin{equation}
      \label{Slater}
      \Psi:= \frac{1}{\sqrt{N!}}
      \begin{vmatrix}
        \xi_1(x_1)&\hdots&\xi_1(x_N)\\
        \vdots&&\vdots\\
        \xi_N(x_1)&\hdots&\xi_N(x_N)
      \end{vmatrix}
    \end{equation}
    for the Slater determinant of these orbitals. A straightforward
    computation shows
    \begin{align}
      &\mathrel{\phantom{=}}
      E^\mathrm{S}(Z,N)\leq (\Psi, H_{Z/|\gx|,N}\Psi)
      = \cE_{Z}^{\mathrm{classical\ HF}}(\xi_1,\ldots,\xi_N)\nonumber\\
      &=\sum_{n=1}^N\int_\Gamma\rd x\,
        {\left(|\grad\xi_n(x)|^2-\frac{Z}{|\gx|}|\xi(x)|^2\right)}
        \nonumber\\
      &\mathrel{\phantom{=}}\mbox{}
      + \frac12\int_{\Gamma^2}\rd x\,\rd y\,
      \frac{\sum_{n=1}^N|\xi_n(x)|^2\sum_{m=1}^N|\xi_m(y)|^2
        -\sum_{n=1}^N\bigl|\xi_n(x)\overline{\xi_n(y)}\bigr|^2}{|\gx-\gy|}
      \nonumber\\
      &= \cE^\mathrm{HF}_Z(\gamma)
    \end{align}
    where we used the Ritz variational principle \cite{Ritz1909} in
    the first step.
  \end{proof}

  \noindent%
  In fact, Bach \cite{Bach1992} proved a correlation bound that allows
  to prove also the reverse inequality (for suitable choice of $N$
  that includes $N=Z$) up to errors yielding
  \begin{equation}
    \label{Bach53}
    E^\mathrm{HF}(Z)= E^\mathrm{TF}(Z)+ \frac{Z^2}{2}-\gamma_SZ^\frac53+o(Z^\frac53).
  \end{equation}
  
  As mentioned already in \eqref{RafaelH} the maximal number of
  electrons of any Hartree--Fock minimizer is bounded by
  $2Z+1$. Solovej \cite{Solovej2003} has improved this by the
  following asymptotic bound.
  \begin{theorem}
    There exists a constant $\const$ such for all $Z$ any minimizer of
    $\cE_Z^\mathrm{HF}$ on $\cD$ fulfills
    \begin{equation}
      \label{Solovej}
      \tr(\gamma)-Z\leq C,
    \end{equation}
    i.e., the charge in excess of neutrality is bounded uniformly in
    the atomic number $Z$.
  \end{theorem}
  The method is to use Thomas--Fermi theory successively to screen the
  nucleus.

  \subsection{The M\"uller functional}
  Comparing with \eqref{dd}, the Hartree--Fock functional can be viewed
  as the full quantum functional with an ansatz for the two particle
  density
  \begin{equation}
    \label{ansatzhf}
    \rho^{(2)}(\gx,\gy)
    := \frac12\sum_{\sigma,\tau=1}^2
    \bigl(\gamma(x,x)\gamma(y,y)-|\gamma(x,y)|^2\bigr).
  \end{equation}
  This ansatz keeps the property that ${\rho^{(2)}\geq0}$, which we
  showed in fact in \eqref{xld} when we dominated the exchange term by
  the classical electron-electron interaction. However, it does, in
  general, not keep the normalization, since
  \begin{align}
    \label{norm}
    \int_{\rz^6}\rho^{(2)}
    &=\int_{\rz^3}\rd\gx\int_{\rz^3}\rd\gy\,
      \frac12\sum_{\sigma,\tau=1}^2
      \bigl(\gamma(x,x)\gamma(y,y)-|\gamma(x,y)|^2\bigr)\nonumber\\
    &=\frac12\left(N^2 -\tr(\gamma^2)\right) \geq \binom{N}{2}
  \end{align}
  where equality holds if and only if $\gamma$ is a projection.

  M\"uller \cite{Muller1984} also makes an ansatz for the two-particle
  density, however reversing the situation. He keeps the normalization
  condition and gives up the positivity. His ansatz is
  \begin{equation}
    \label{ansatzM}
    \rho^{(2)}(\gx,\gy)
    := \frac12\sum_{\sigma,\tau=1}^2
    \bigl(\gamma(x,x)\gamma(y,y)-|\gamma^\frac12(x,y)|^2\bigr)
  \end{equation}
  where the root denotes the operator root, i.e., the eigenvalues of
  $\gamma$ are replaced by their roots. (Actually M\"uller considered
  a family of functionals with the exchange part of the two-particle
  density depending on parameter $p\in[0,\tfrac12]$. The functional
  considered here is the case $p=0$.) Thus the M\"uller functional
  reads
  \begin{align}
    \label{Muller}
    &\cE_Z^\mathrm{M}:\cD\to\rz,\nonumber\\
    &\gamma\mapsto \tr\biggl(-\frac12\Delta\gamma\biggr)
    -\int_{\rz^3}\rd \gx\,\frac{Z\rho_\gamma(\gx)}{|\gx|}+ D[\rho_\gamma]
    -X[\gamma^\frac12].
  \end{align}
  To show that the functional is well defined it suffices to show that
  $X[\gamma^\frac12]$ is finite, since the first three terms are
  identical with the first three terms of the Hartree--Fock functional.
  \begin{align}
    X[\gamma^{\frac12}]
    &\leq \frac12\int_{\Gamma}\rd y\, \rd x\,
      \frac{|\gamma^\frac12(x,y)|^2}{|\gx-\gy|}\nonumber\\
    &\leq \frac12\sqrt{\int_{\Gamma^2} \rd y\, \rd x\,|\gamma^\frac12(x,y)|^2}
    \sqrt{\int_{\Gamma^2} \rd y \,\rd x\,
    \frac{|\gamma^\frac12(x,y)|^2}{|\gx-\gy|^2}}\nonumber\\
    &\leq \sqrt{\tr(\gamma)}\sqrt{\tr(-\Delta\gamma)}<\infty
  \end{align}
  using the Schwarz inequality first and then, in the last step,
  Hardy's inequality and the fact that
  $\int_{\Gamma^2}\rd x\,\rd y\,|\gamma^\frac12(x,y)|^2= \tr(\gamma)$. The
  functional has been analyzed by Frank et
  al.\ \cite{Franketal2007}. One of the interesting properties of the
  functional is its convexity which follows from an entropy inequality
  of Wigner and Yanase \cite{WignerYanase1963}. On the other hand it
  suffers from the same defect as the Thomas--Fermi--Dirac functional,
  the TF functional minus a $\int\rho^\frac43$-term: Even if $Z=0$ the
  infimum of the M\"uller functional is negative. Thus the binding
  energy of the electrons is not just the infimum of the M\"uller
  functional but the infimum at $Z=0$, which turns out to be
  $\tr(\gamma)/8$, has to be subtracted. We write $\hat\cE^\mathrm{M}_Z$
  for the correspondingly modified functional.

  It is $\hat\cE^\mathrm{M}_Z$ that determines the maximal number of
  electrons that can be bound. The critical electron number
  $N_c^\mathrm{M}$ --- defined as the largest electron number $N$ for
  which $\hat\cE_Z^\mathrm{M}$ ceases to have a minimizer on
  $\cD_{\partial N}$ --- is at least $Z$. On the other hand Frank et
  al.\ \cite{Franketal2018T} showed that, like in the Hartree--Fock
  case, there is a constant $\const$ such that
  \begin{equation}
    \label{excessM}
    N_c^\mathrm{M}-Z\leq \const
  \end{equation}
  uniformly in $Z$. They use a different technique from the one used
  in \eqref{mulmitx} to obtain the necessary a priori bound followed
  by a variation of Solovej's successive screening. The technique
  works also in the case of exchange terms that are not dominated by
  the classical electron-electron interaction like the
  Thomas--Fermi--Dirac--Weizs\"acker functional where they developed
  their technique (Frank et al.\  \cite{Franketal2018}) (see also Chen et
  al.\ \cite{Chenetal2020} for the a priori bound for the relativistic
  density functional of Engel and Dreizler). The bound \eqref{excessM}
  on the excess charge has been generalized by Kehle \cite{Kehle2017}
  to a functional of Sharma et al.\  \cite{Sharmaetal2008} with
  $X[\gamma^\frac12]$ replaced by $X[\gamma^p]$ with $p\in[\frac12,1]$
  which interpolates between the M\"uller and the Hartree--Fock
  functional.

  Energetically the same accuracy as for the Hartree--Fock functional
  in \eqref{Bach53} is known, i.e.,
  \begin{equation}
    \label{S53}
    E^{\mathrm{M}}(Z)= E^{\mathrm{TF}}(Z)+ \frac{Z^2}{2}
    -\gamma_SZ^{\frac53}+o(Z^{\frac53})
  \end{equation}
  which is shown by proving
  $E^{\mathrm{M}}(Z)=E^{\mathrm{HF}}+o(Z^{\frac53})$ (see
  \cite{Siedentop2009,Siedentop2014}).

  Thus, as far as the current knowledge of the asymptotics of the
  ground-state energy is concerned, the M\"uller functional yields the
  Hartree--Fock functional accuracy. There are however indications that
  the two complement each other also in a different way: while the
  Hartree--Fock functional is known to yield an upper bound on the true
  quantum energy, it is conjectured that the ground-state energy of
  the M\"uller functional is a lower bound to the quantum molecular
  ground-state energy. Frank et al.\  \cite[Section V]{Franketal2007}
  proved this for $N=2$. Moreover, numerical evidence suggests that it
  might be true for all $N$. A positive answer to this question would
  give an even more advanced tool to estimate the energy of Coulomb
  systems from below than the Lieb--Thirring inequality with the
  conjectured classical constant.

\setcounter{equation}{0}\renewcommand{\theequation}{A.\arabic{equation}}
\setcounter{theorem}{0}\renewcommand{\thetheorem}{A.\arabic{theorem}}
  
\section*{Appendix: %
    Maximal functions of powers and Thomas--Fermi energy
    of exchange holes}%
\addcontentsline{toc}{section}{Appendix: Maximal functions of powers and
  Thomas--Fermi energy of exchange holes}%
First we show that inverse powers $|\gx|^{-\alpha}$ are eigenfunctions
of the maximal operator with eigenvalue
$C_{\alpha,d}:=M(|\cdot|^{-\alpha})((0,0,1))$ or less poetical
\begin{lemma}
  \label{max}
  Assume $\alpha\in [0,d)$. Then
  \begin{equation}
    \label{mf}
    M(|\cdot|^{-\alpha})= C_{\alpha,d}|\cdot|^{-\alpha}.
  \end{equation}
\end{lemma}\clearpage%%++
  \begin{proof}
    We write $\omega_d$ for the volume of the unit ball in $d$
    dimensions, pick $\gx\in\mathbb{S}^2$, and compute
    \begin{multline}
      M(|\cdot|^{-\alpha})(\gx) := \sup_{R>0} \frac{\int_{|\gx-\gy|<R}\rd
        \gy |\gy|^{-\alpha}}{\omega_dR^d}
      = \sup_{R>0} \frac{\int_{||\gx|\ge-\gy|<R}
        \rd \gy |\gy|^{-\alpha}}{\omega_dR^d}\\
      = \sup_{R>0}\frac{|\gx|^{d-\alpha}\int_{||\gx|\ge-|\gx|\gy|<|\gx|R}\rd \gy
        |\gy|^{-\alpha}}{\omega_d|\gx|^dR^d}= |\gx|^{-\alpha}
      \sup_{R>0} \frac{\int_{|\ge-\gy|<R}\rd \gy
        |\gy|^{-\alpha}}{\omega_dR^d}
    \end{multline}
    where we observe in the first step that the integral does not
    depend on the direction of $\gx$, and in the second step we
    scale $\gy\to |\gx|\gy$ and $R\to |\gx|R$.
  \end{proof}
  
  Next we turn to the infimum of the Thomas--Fermi functional with root
  of the Coulomb potential $(Z/|\gx|)^\frac12$, arbitrary positive
  Thomas--Fermi constant $\gamma$, and constraint $\int\rho\leq Z$ on
  the electron number.
  \begin{lemma}
    For ${\gamma,Z,C\in\rz_+}$ and ${M_Z:=\{\rho\in L^\frac53(\rz^3)|
    \rho\geq0, \int_{\rz^3}\rho\leq Z\}}$ set
 \begin{equation}\label{defI}
   I_{\gamma,Z}:=  \inf_{\rho\in M_Z}\int_{\rz^3}\rd\gx
   \left(\frac35\gamma\rho(\gx)^\frac53 -
      C\sqrt{\frac{Z}{|\gx|}}\rho(\gx)\right).
  \end{equation}
  Then
  \begin{equation}
    \label{I}
    I_{\gamma,Z}= \frac{Z^\frac{13}9}{\gamma^\frac13}I_{1,1}.
  \end{equation}
\end{lemma}
\begin{proof}
  By scaling $\rho\rightarrow Z\alpha^3\rho(\alpha \cdot)$ we get
  \begin{equation}
      I_{\gamma,Z}=
      \frac{1}{\alpha^3}\inf_{\rho\in M_1}\int_{\rz^3}\rd\gx
      \left(\frac35\gamma Z^\frac53\alpha^5\rho(\gx)^\frac53 -
      CZ^\frac12\alpha^\frac12Z\alpha^3\sqrt{\frac1{|\gx|}}\rho(\gx)\right).
  \end{equation}
  Picking ${\alpha := \gamma^{-\frac23} Z^{-\frac19}}$ implies
  ${\gamma Z^\frac53\alpha^5=Z^\frac32\alpha^\frac72}$ which allows to
  take this common factor in front of the infimum yielding eventually
  \begin{equation}
    I_{\gamma,Z}= \frac{Z^\frac{13}9}{\gamma^\frac13}
    \inf_{\rho\in M_1}
    \int_{\rz^3}\rd\gx\left(\frac35\rho(\gx)^\frac53 -
      C\frac1{|\gx|^\frac12}\rho(\gx)\right);
  \end{equation}
  quod erat demonstrandum.
\end{proof}

\section*{Acknowledgments}\addcontentsline{toc}{section}{Acknowledgments}
Thanks go to Hongshuo Chen and Konstantin Merz for a critical
reading of the manuscript. Partial support by the Deutsche
Forschungsgemeinschaft (DFG, German Research Foundation) through
Germany's Excellence Strategy EXC-2111-390814868 is gratefully
acknowledged.
Special thanks go the Institute of Mathematical Sciences at the
National University of Singapore for support through the programme
\textit{Density Functionals for Many-Particle Systems: Mathematical
  Theory and Physical Applications of Effective Equations}.

%\bibliographystyle{plain}
%\bibliography{shortJournalNames,coulomb}

\def\cprime{$'$}

\end{document}